%% file: STL_Comparison.tex
\documentclass[letterpaper, 10 pt, conference]{ieeeconf}  

\IEEEoverridecommandlockouts                              
\input{def}
\usepackage{scalefnt}

\usepackage[ruled]{algorithm2e} 

\SetAlFnt{\small}
\SetAlCapFnt{\small}
\SetAlCapNameFnt{\small}
\SetAlCapHSkip{0pt}
\IncMargin{-\parindent}

\newcommand{\journal}[1]{}

\usepackage{siunitx} 
\DeclareSIUnit\Molar{\textsc{m}}

\usepackage{verbatim}
\usepackage{soul}
\usepackage{cite}
\usepackage{balance}

\usepackage{color, colortbl}

\newtheorem{define}{Definition}
\newtheorem{lemma}{Lemma}
\newtheorem{example}{Example}

\newtheorem{theorem}{Theorem}
\newtheorem{corollary}{Corollary}
\newtheorem{proposition}{Proposition}

\title{\LARGE \bf
Metrics for Signal Temporal Logic Formulae
}

\author{Curtis Madsen$^*$, Prashant Vaidyanathan$^*$, Sadra Sadraddini$^*$, Cristian-Ioan Vasile$^*$, Nicholas A. DeLateur,\\Ron Weiss, Douglas Densmore, and Calin Belta
\thanks{$^*$These authors contributed equally. This work was partially supported by the National Science Foundation under grant CPS Frontier 1446607 and NSF IIS-1723995. Curtis Madsen (ckmadsen@bu.edu), Prashant Vaidyanathan (prash@bu.edu), Douglas Densmore (dougd@bu.edu), and Calin Belta (cbelta@bu.edu) are at Boston University, Boston, MA, USA. Sadra Sadraddini (sadra@mit.edu), Cristian-Ioan Vasile (cvasile@mit.edu), Nicholas A. DeLateur (delateur@mit.edu), and Ron Weiss (rweiss@mit.edu) are at Massachusetts Institute of Technology, Cambridge, MA, USA.  
}
}

\begin{document}

\maketitle
\thispagestyle{empty}
\pagestyle{empty}


\begin{abstract}
Signal Temporal Logic (STL) is a formal language for describing a broad range of real-valued, temporal properties in cyber-physical 
systems.
While there has been extensive research on verification and control synthesis from STL requirements, there is no formal framework for comparing two STL formulae.
In this paper, we show that under mild assumptions, STL formulae admit a metric space.
We propose two metrics over this space based on i) the Pompeiu-Hausdorff distance and ii) the symmetric difference measure, and present algorithms to compute them.
Alongside illustrative examples, we present 
applications of these metrics 
for two fundamental problems: a) {\em design quality measures}: to compare all the temporal behaviors of a designed system, such as a synthetic genetic circuit, with the ``desired'' specification, and b) {\em loss functions}: to quantify errors in Temporal Logic Inference (TLI) as a first step to establish formal performance guarantees of TLI algorithms.

\end{abstract}






\input{introduction}

\input{preliminaries}

\input{metrics}

\input{methods}

\input{case_study}

\input{discussion}


\bibliographystyle{IEEEtran}
\bibliography{references}

\balance

\end{document}

%% file: def.tex
\usepackage{amsmath}
\usepackage{amssymb}
\usepackage{amsfonts}
\usepackage{float}
\usepackage{verbatim}
\usepackage{multirow}
\usepackage{graphicx}
\usepackage{color}
\usepackage{url}
\usepackage{xspace} 
\usepackage{caption}
\usepackage{hyperref}
\usepackage[caption=false,font=footnotesize]{subfig}

\usepackage[vlined,ruled]{algorithm2e}
\usepackage{multirow} 
\usepackage{tikz,pgf}
\usetikzlibrary{arrows,automata,shapes,calc,backgrounds,spy}
\usetikzlibrary{external}
\usepackage{diagbox}

\usepackage[colorinlistoftodos,prependcaption,textsize=tiny]{todonotes}
\usepackage{marginnote}
\makeatletter
\renewcommand{\@todonotes@drawMarginNoteWithLine}{%
\begin{tikzpicture}[remember picture, overlay, baseline=-0.75ex]%
    \node [coordinate] (inText) {};%
\end{tikzpicture}%
\marginnote[{
    \@todonotes@drawMarginNote%
    \@todonotes@drawLineToLeftMargin%
}]{
    \@todonotes@drawMarginNote%
    \@todonotes@drawLineToRightMargin%
}%
}
\makeatother

\newtheorem{assumption}{Assumption}

\newcommand{\CA}[1]{\mathcal{#1}}

\newcommand{\BB}[1]{\mathbb{#1}}

\newcommand{\set}{\mathcal{S}}

\newcommand{\notltl}{\neg}
\newcommand{\andltl}{\land}
\newcommand{\orltl}{\lor}

\newcommand{\Always}{\Box}
\newcommand{\Event}{\diamondsuit}
\newcommand{\Until}{\xspace\CA{U}\xspace}

\newcommand{\True}{\top}
\newcommand{\False}{\perp}

\newcommand{\Proj}{\CA{P}}

\newcommand{\mudist}{distance\xspace}

\newcommand{\hausdorff}{PH\xspace}

\newcommand{\symdiff}{SD\xspace}

\newcommand{\maxhorz}{T}

\newcommand{\stldist}{d_{STL}}

\newcommand{\buchi}{B\"uchi\ }

\newcommand{\stl}{STL\xspace}

\newcommand{\eqclass}[1]{\left\langle {#1} \right\rangle}
\newcommand{\norm}[1]{\left\| {#1} \right\|}

\newcommand{\spow}[1]{2^{#1}}

\DeclareMathOperator*{\argmin}{arg\!min}

%% file: introduction.tex
\section{Introduction}

Temporal logics~\cite{baier2008modelchecking} are increasingly used for describing specifications in cyber-physical systems such as robotics~\cite{kress2009temporal}, synthetic biology~\cite{batt2007robustness}, and transportation ~\cite{coogan2017formal}.
Variants of temporal logics, such as {\em Computation Tree Logic} (CTL)~\cite{Emerson1982}, {\em Linear Temporal Logic} (LTL)~\cite{pnueli1977temporal}, or  {\em Signal Temporal Logic}~(STL)~\cite{maler2004monitoring}, can naturally describe a wide range of temporal system properties such as \emph{safety} (never visit a ``bad" state), \emph{liveness} (eventually visit a ``good" state), \emph{sequentiality}, and their arbitrarily elaborate combinations.

Using {\em model checking}~\cite{baier2008modelchecking} techniques, signals or traces can be checked to determine whether or not they satisfy a specification.
For STL in particular, the {\em degree of satisfaction} or \emph{robustness} is a quantitative measure to characterize how far a signal is from satisfaction~\cite{fainekos2009robustness,donze2010robust,donze2013efficient} of an STL formula.
There is currently, however, no formal way to directly compare specifications against each other.
Previous related approaches in planning and control have looked into specification
relaxation, where the goal is to minimally enlarge the specification language to
include a satisfying control policy for the system model.
Various specification relaxations have been defined including minimum
violation~\cite{Tumova-ACC13,Tumova-HSCC13,VaTuKaBeRu-ICRA-2017} for self-driving
cars, temporal relaxation of deadlines~\cite{VaAkBe-TCS-2016},
minimum revision of \buchi automata~\cite{kim2015minimalrevision}, and 
diagnosis and repair in reactive synthesis~\cite{ghosh2016diagnosis}.
While language inclusion and equivalence problems are of paramount importance in computer science and control theory, they are only qualitative measures while we are interested in quantitative metrics. 

This paper presents two metrics that can be used to compute the \mudist between two STL specifications. Under mild assumptions, we propose the metrics based on the languages of STL formulae. We propose two distance functions. The first is based on the \emph{Pompeiu-Hausdorff} (\hausdorff) distance \cite{munkres2000topology}, which captures how much the language of one formula must be enlarged to include
the other, and the second is based on the \emph{symmetric difference} (\symdiff) \cite{conway2012course}, which characterizes how much overlap there is between the two formulae.
The theoretical contributions of this paper are: 
\begin{enumerate}
\item[1.] formalization of STL formulae metrics based on the \hausdorff and the \symdiff distances, and
\item[2.] methods for computing the \hausdorff using \emph{mixed-integer linear programming} (MILP), and the \symdiff using a recursive algorithm based on the \emph{area of satisfaction}.
\end{enumerate}
We discuss the comparison of the two metrics in detail and provide examples that highlight their differences.

This paper additionally presents applications of the proposed metrics to a behavioral synthesis problem and to the evaluation of {\em temporal logic inference} (TLI)~\cite{Bartocci2014,Jin2015,bombara_decision_2016,vaidyanathan2017grid,Hoxha2018} methods.
In the first case, we are interested in generating designs that exhibit desired behaviors specified in STL.
For example, we study synthetic genetic circuits.
Possible circuit designs are constructed and measured in laboratory experiments, and the resulting traces are abstracted into STL specifications using TLI.
These formulae are compared quantitatively against the desired design specification using the proposed metrics.
The second setup considers the fundamental problem of evaluating TLI methods themselves.
Under the assumption that data used for inference can be characterized by ground truth STL formulae, we ask the question of how well the TLI algorithms perform.
As opposed to empirical evaluation used in previous work, we propose to use our metrics as loss functions as the first step in establishing theoretical foundations for TLI.
The related contributions are:
\begin{enumerate}
\item[3.] a design quality measure for evaluating proposed implementations against an STL specification, and 
\item[4.] a loss function to quantify errors in TLI as a first step in establishing formal performance guarantees of TLI algorithms.
\end{enumerate}

%% file: preliminaries.tex
\section{Preliminaries}
\label{sec:preliminaries}

\subsection*{Notation}
Let $\BB{R}, \BB{R}_{\ge 0}$, $\BB{N}$ denote the set of real, non-negative real, and natural numbers, respectively. We use $|r|$ to denote the absolute value of $r\in \BB{R}$. Given $x \in \BB{R}^n$, $\left\|x\right\|_\infty:=\max_{i \in \{1,\cdots,n\}} |x^i|$ - where $x^i$ is the $i$'th component of $x$ - is its infinity-norm. A scalar-valued function $f: \BB{R}^n \rightarrow \BB{R}$ is \emph{rectangular} if for some $i \in \{1,\cdots,n\}$, $f(x)=x^i$. 
A metric space is an ordered pair $(\CA{M}, d)$, where $\CA{M}$ is a set and $d: \CA{M} \times \CA{M} \rightarrow \BB{R}_{\ge 0}$ is a distance function such that i) $d(x,y)=0 \Leftrightarrow x=y$; ii) $d(x,y)=d(y,x), \forall x,y \in \CA{M};$ iii) $d(x,z) \le d(x,y)+d(y,z), \forall x,y,z \in \CA{M}$. If $(\CA{M},d_1)$ and $(\CA{M},d_2)$ are metric spaces, then $(\CA{M},\lambda d_1 + (1-\lambda) d_2)$, is also a metric space for any $0 \le \lambda \le 1$.
 

We use discrete notion of time throughout this paper. Time intervals in the form $I=[t_1,t_2] \subset \BB{N}$, $t_1,t_2 \in \BB{N}, t_1\le t_2$, are interpreted as $\{t_1,t_1+1,\cdots,t_2\}$. $[\tau+t_1, \tau+t_2]$ is denoted by $\tau+I$, $\tau \in \mathbb{N}$. The continuous interval $\{r|0 \le r \le 1\}$ is denoted by $\mathbb{U}$. 
An $n$-dimensional, real, infinite-time, discrete-time signal $s$ is defined as a string of real values $s:s_0s_1 s_2\cdots $, where $s_t \in \BB{S}, \BB{S} \subset \BB{R}^n$, $t \in \BB{N}$. The \emph{suffix} of $s$ at $t$, denoted by $s[t]$, is a signal such that $s[t]_\tau = s_{t+\tau}$ for all $t,\tau \in \BB{N}$. We use $s[t_1,t_2]:=s_{t_1}s_{t_1+1}\cdots s_{t_2}$ to refer to a particular portion of a signal. The set of all signals with values taken in $\BB{S}$ is denoted by $\CA{S}$. The set of all signal prefixes with time bound $T$ is defined as 
$
\CA{S}_T:=\left\{ s[0:T] \mid s \in \CA{S} \right\}.
$
For the convenience of notation, we use $s\in \CA{S}_T$ to say $s[0:T] \in \CA{S}_T$. 
The distance between two signals $s,s' \in \CA{S}_T$ is defined as $
\label{eq:signal-dist}
d(s, s') := \sup_{t \in [0,T]} \left\{ \|s_t-s'_t\|_\infty \right\}.
$

\subsection*{Signal Temporal Logic}

The syntax of STL is defined as follows~\cite{maler2004monitoring}:
\begin{equation*}
\phi ::= \True \ |\  \pi \ |\ \notltl \phi  \ |\ \phi_1 \andltl \phi_2 \ |\ \phi_1 \Until_{I} \phi_2 \ ,
\end{equation*}
where $\True$ is the Boolean {\em true} constant; $\pi$ is a predicate over $\BB{R}^n$ in the form of 
$f(x) \sim \mu$,
$f: \BB{S} \rightarrow \BB{R}$, $\mu \in \BB{R}$, and $\sim \in \{\le,\ge\}$; $\neg$ and $\wedge$ are the Boolean operators for negation and conjunction, respectively; and $\Until_{I}$ is the temporal operator {\em until} over bounded interval $I$. A predicate $f(s)\sim \mu$ is rectangular if $f$ is rectangular. Other Boolean operations are defined in the usual way. Additional temporal operators {\em eventually} and {\em globally} are defined as $\Event_{I} \phi \equiv \True \Until_{I} \phi$ and $\Always_{I} \phi \equiv \notltl \Event_{I} \notltl \phi$,  respectively, where $I$ is an interval. The set of all STL formulae over signals in $\CA{S}$ is denoted by $\Phi^{\CA{S}}$. 
%
The \emph{STL score}, also known as \emph{robustness degree} is a function $\rho: \CA{S} \times \Phi^{\CA{S}} \times \BB{N} \rightarrow \BB{R}$, which is recursively defined as \cite{maler2004monitoring}: 
\begin{equation}
\label{eq:quantitative-semantics}
\begin{array}{ll}
\rho(s,(f(s) \sim \mu),t) & = \begin{cases} \mu-f(s_t) & \sim=\le \\ f(s_t)-\mu & \sim=\ge \end{cases},
\\
\rho(s,\neg \phi,t)   & = -\rho(s,\phi,t),
\\
\rho(s,\phi_1 \orltl \phi_2,t)   & = \max(\rho(s,\phi_1,t),\rho(s,\phi_2,t) ),
\\
\rho(s,\phi_1 \andltl \phi_2,t)   & = \min(\rho(s,\phi_1,t),\rho(s,\phi_2,t) ),
\\
\rho(s,\phi_1~{\Until}_{I}~ \phi_2,t) & = \underset{t^\prime \in t+I} \max \big (  
 \rho(s,\phi_2,t'), \\
 &~~~~~~~ \underset{t'' \in [t,t']} \min \rho(s,\phi_1,t'')\big),
\\
\rho(s,\Event_{I}~\phi,t) & = \underset{t' \in t+I}\max ~ \rho(s,\phi,t'), \\
\rho(s,\Always_{I}~\phi,t) & = \underset{t' \in t+I}\min~  \rho(s,\phi,t').
\end{array}
\end{equation}
As one can inspect from~\eqref{eq:quantitative-semantics}, a signal \emph{satisfies} an STL specification at a certain time if and only if its corresponding STL score is positive: $
s[t] \models \phi \Leftrightarrow \rho(s,\phi,t)>0, 
$ where $\models$ is read as ``satisfies". The case of $\rho=0$ is usually left ambiguous - this is never a concern in practice due to issues with numerical precision. In this paper, we consider $\rho=0$ as satisfaction, but by doing so, we sacrifice the principle of contradiction: $s[t] \models \phi$ and $s[t] \models \neg \phi$ if $\rho(s,\phi,t)=0$. 


The horizon of an STL formula is defined as the minimum length of the time window required to compute its score, and it is recursively computed as \cite{Dokhanchi2014}:
\begin{equation}
\label{eq_horizon} 
\begin{array}{l}
\norm{\pi}=0, \norm{\phi}=\norm{\neg \phi}\\
\norm{\phi_1 \andltl \phi_2}=\norm{\phi_1 \orltl \phi_2}= \max\{\norm{\phi_1}, \norm{\phi_2}\} \\
\norm{\phi_1 \Until_{[t_1,t_2]} \phi_2}=t_2+ \max\{\norm{\phi_1}, \norm{\phi_2}\}\\
\norm{\Event_{[t_1, t_2]} \phi}=\norm{\Always_{[t_1, t_2]} \phi}=t_2+ \norm{\phi}\\
\end{array}
\end{equation}
The set of all STL formulae over signals in $\CA{S}$ such that their horizons are less than $T$ is denoted by $\Phi^{\CA{S}_T}$. 
Note that computing $\rho(s,\phi,t)$ requires $s[t:t+\norm{\phi}]$, and the rest of the values are irrelevant.  
\begin{define}[Bounded-Time Language]
Given $\phi \in \Phi^{\CA{S}_T}$, we define the bounded-time language as:
\begin{equation}
\label{eq:def-language}
\CA{L}(\phi):=\left\{ s \in \CA{S}_T \mid \rho(s,\phi,0) \ge 0 \right\}.
\end{equation}
\end{define}
Note that $\CA{L}(\phi) \subset \BB{R}^{n(T+1)}$. When the predicates are rectangular, the bounded-time language becomes a finite union of hyper-rectangles in $\BB{R}^{n(T+1)}$. 

\begin{example}
Let $\BB{S}= \BB{U}$. Consider the following six STL formulae in $\Phi^{\CA{S}_{20}}$:
\begin{equation}
\label{eq_formulas}
\begin{array}{c}
\phi_1=\Always_{[0,20]} \theta_1, \phi_2=\Always_{[0,20]} \theta_2,  \\
\phi_3=\Event_{[0,20]} \theta_1,
\phi_4=\Always_{[0,20]} \theta_1 \wedge \Event_{[0,20]} \theta_2, \\
\phi_5=(\Always_{[0,10]} \theta_1) \wedge (\Always_{[12,20]} \theta_2), \phi_6=\Always_{[0,16]} \Event_{[0,4]} \theta_1, \\
\end{array}
\end{equation}
where $\theta_1=(x \ge 0.2) \wedge (x \le 0.4), $ and $\theta_2=(x \ge 0.2) \wedge (x \le 0.44)$. We have $\norm{\phi_i}=20, i=1,\cdots,6$. Two examples of bounded-time languages are: $
\CA{L}(\phi_2)=\bigcap_{\tau=0}^{20} \{0.2 \le x_t \le 0.44 \},
\CA{L}(\phi_3)=\bigcup_{\tau=0}^{20} \{0.2 \le x_t \le 0.4 \}.
$ 
Consider two constant signals $s^1$ and $s^2$, where $s^1_t=0.3, s_t^2=t/20, t=0,1,\cdots,20$. The STL scores are computed from \eqref{eq:quantitative-semantics}. For instance, $\rho(s^1,\phi_1,0)=0.1$, $\rho(s^2,\phi_1,0)=-0.6$ (minimizer at $t=20$), and $\rho(s^2,\phi_3,0)=0.1$ (maximizer at $t=6$).
\end{example}

%% file: metrics.tex
\section{Metrics}
\label{sec:metrics}

In this section, we introduce two functions
$\stldist: \Phi^{\CA{S}_T} \times \Phi^{\CA{S}_T} \rightarrow \BB{R}_{\ge 0}$
that quantify the dissimilarity between the properties
captured by the two STL formulae. However, it \emph{is} possible that different formulae may describe the same properties. For example, $\phi_1$ and $\phi_4$ in \eqref{eq_formulas} are describing the same behavior, since any signal that satisfies $\phi_1$ already satisfies $\phi_4$ \emph{and} vice versa. The key idea is to define the distance between two STL formulae as the distance
between their time-bounded languages. 
\begin{assumption}
The set $\BB{S} \subset \BB{R}^n$ is compact. 
\end{assumption}
\begin{assumption}
\label{assume_rectangle}
All of the predicates are rectangular.
\end{assumption}
Note that bounded-time languages are constructed in finite-dimensional Euclidean spaces. Also, since all inequalities in the predicates are non-strict, bounded-time languages are compact sets. Assumption \ref{assume_rectangle} is theoretically restrictive, but not in most applications - usually it is the case that all predicates are rectangular as they describe thresholds for state components of a system.

\begin{define}
We say that the two STL formulae $\phi_1$ and $\phi_2$ are \emph{semantically equivalent},
denoted by $\phi_1 \equiv \phi_2$, if both induce the same language: $\CA{L}(\phi_1) = \CA{L}(\phi_2)$.
\end{define}
The set of equivalence classes of $\Phi^{\CA{S}_T}$ induced by $\equiv$ is denoted by $\Phi^{\CA{S}_T}/\equiv$.
Distance functions $\stldist$ are effectively pseudo-metrics on $\Phi^{\CA{S}_T}$,
but proper metrics on $\Phi^{\CA{S}_T}/\equiv$, where
$\stldist(\eqclass{\phi_1}, \eqclass{\phi_2}) = \stldist(\phi_1, \phi_2)$
is the induced metric,
$\eqclass{\phi}$ is the equivalence class associated with $\phi$,
and $\phi_1$ and $\phi_2$ are formulae in the two equivalence classes.
Note that, by definition, there is a one-to-one map between the
equivalence classes of STL formulae and their formulae.
Moreover, for any $\phi_1, \phi_2 \in \eqclass{\phi}$, we have  $\stldist(\phi_1, \phi_2) = 0$.

We adapt two common metrics between sets:
(a) the {\em Pompeiu-Hausdorff} (\hausdorff) {\em distance} based on the underlying metric between
signals,
and (b) a measure of {\em Symmetric Difference} (\symdiff) between sets. As it will be clarified in the paper, the choice of $T$, as long as it is larger than the horizons of the formulae that are considered, does not affect the fundamental properties of the defined metrics. In the case of the \hausdorff distance, it does not have any effect at all. For the \symdiff metric, the computed distances are scaled with respect to the inverse of $T$. These details are explained in Section~\ref{sec:sd-dist-theory}.

%
%

\subsection{Pompeiu-Hausdorff Distance}
\label{sec:ph-dist-theory}
\begin{define}
The (undirected) \hausdorff distance is defined as:
\begin{equation}
d_{\hausdorff}(\phi_1,\phi_2)=\max\left\{\vec{d}_{\hausdorff}(\phi_1,\phi_2),\vec{d}_{\hausdorff}(\phi_2,\phi_1) \right\},
\end{equation}
where $\vec{d}_{\hausdorff}$ denotes the directed \hausdorff distance: 
\begin{equation}
\label{eq_directed}
\vec{d}_{\hausdorff}(\phi_1,\phi_2):=\sup_{s_1 \in \CA{L}(\phi_1)} \left\{ \inf_{s_2 \in \CA{L}(\phi_2)} d(s_1,s_2) \right\}.
\end{equation}
\end{define}
Note that the directed \hausdorff distance is obviously not a metric as it is possible to have $\vec{d}_{\hausdorff}(\phi_1,\phi_2) \neq \vec{d}_{\hausdorff}(\phi_2,\phi_1)$. We have $\vec{d}_{\hausdorff}(\phi_1,\phi_2)=0$ if and only if $\mathcal{L}(\phi_1) \subseteq \mathcal{L}(\phi_2)$. Another way to interpret the \hausdorff distance is as follows \cite{munkres2000topology}:
\begin{equation}
\label{PH_ball}
\vec{d}_{\hausdorff}(\phi_1,\phi_2)=\min\{\epsilon \mid \CA{L}(\phi_1) \subseteq \CA{L}(\phi_2)+\epsilon \CA{B}^{\CA{S}_T} \},
\end{equation}
where $\CA{B}^{\CA{S}_T}$ is the unit ball in ${\CA{S}_T}: \{s[0:T] \mid \left\|s_t\right\|_\infty \le 1, t \in [0,T] \},$ and addition of sets is interpreted in the Minkowski sense. 
In words, $\vec{d}_{\hausdorff}(\phi_1,\phi_2)$ is the radius of the minimum ball that should be added to $\CA{L}(\phi_2)$ such that it contains $\CA{L}(\phi_1)$.
\begin{proposition}
The $(\Phi^{\CA{S}_T}/\equiv, d_{\hausdorff})$ is a metric space.
\end{proposition}
\begin{proof}
Note that $d_{\hausdorff}(\phi_1,\phi_2)$ is effectively defined as $d_{\hausdorff}(\CA{L}(\phi_1),\CA{L}(\phi_2))$ - remember that languages are compact subsets of finite dimensional Euclidean space, for which it is known that the \hausdorff distance is a metric \cite{munkres2000topology}. Moreover, there is a one to one map between an equivalency class formula of a formula in $\Phi^{\CA{S}_T}/\equiv$ and its language. 
\end{proof}

It is possible to interpret \eqref{eq_directed} as the distance between an STL formula and a signal: $ 
\vec{d}_{\hausdorff}(s,\phi):= \min_{s' \in \CA{L}(\phi)} d(s,s')$. It is easy to see that we have $\vec{d}_{\hausdorff}(s,\phi)=0$ if and only if $s\in \CA{L}(\phi)$. The following result is a reformulation of Definition~23 in \cite{fainekos2009robustness}, which establishes a connection between the STL score and the notion of \emph{signed distance}. 
\begin{proposition} 
Given any $\phi \in \Phi^{\CA{S}_T}$ and $s \in {\CA{S}_T}$, the STL score is a signed distance in the sense that:
\begin{equation}
\rho(s,\phi,0) = \left\{ \begin{array}{ll} 
-\vec{d}_{PH}(s,\phi) & \vec{d}_{PH}(s,\phi)>0, \\
\vec{d}_{PH}(s,\neg \phi) & \vec{d}_{PH}(s,\phi)=0.
\end{array} \right.
\end{equation}
\end{proposition}  
The following results are extensions of classical results for signed distances \cite{Kraft2015}.
\begin{corollary}
\label{corollary_signed}
For any given two formulae $\phi_1,\phi_2 \in \Phi^{\CA{S}_T}$ and a signal $s\in \CA{S}_T$, we have the following inequalities:
\begin{equation*}
\begin{array}{c}
\big||\rho(s,\phi_1,0)|-|\rho(s,\phi_2,0)|\big| \le d_{\hausdorff}(\phi_1,\phi_2) \\
\!\!\!\big|\rho(s,\phi_1,0)-\rho(s,\phi_2,0)\big| \le d_{\hausdorff}(\phi_1,\phi_2)+d_{\hausdorff}(\neg \phi_1,\neg \phi_2)
\end{array}
\end{equation*}
\end{corollary}
\begin{corollary}
\label{corollary_neighborhood}
Given $\epsilon>0$, define $\epsilon$-neighborhood of an STL formula $\phi$ as $\{\phi\}_\epsilon=\left\{\phi' \in \Phi^{\CA{S}_T} | d_{\hausdorff}(\phi,\phi') \le \epsilon \right\}$. Then, $\rho(s,\phi,0) \ge \epsilon$ implies that $s \models \phi', \forall \phi' \in \{\phi\}_\epsilon$.
\end{corollary}

\subsection{Symmetric Difference}
\label{sec:sd-dist-theory}

The \symdiff is denoted by $\triangle$, and defined as $X\triangle Y = (X\setminus Y) \cup (Y\setminus X)$,
where $X$ and $Y$ are two sets.
It induces a distance between compact sets as the measure of the \symdiff~\cite{conway2012course}.

\journal{
As we mentioned before, we avoid defining a measure over the function space of signals $\CA{S}$, and instead we project them into the space-time set $\BB{S} \times [0, \norm{\phi})$ for some given STL formula $\phi$.
%
%
Next, we define the coverage of an STL formula, and characterize
it for a general case of practical interest.

Let $\Proj: \Phi^{\CA{S}} \to \BB{S} \times \BB{T}$ map STL formulae
to the sets of space-time values that satisfying signals cover.
Formally, we have
$\Proj(\phi) = \bigcup_{s\models \phi} \bigcup_{t \in \left[0, \norm{\phi}\right)} \{(s(t), t)\}$.

\begin{define}
\label{def:symm-diff-dist}
The \symdiff pseudo-metric is defined as:
\begin{equation*}
d_{\symdiff}(\varphi_1, \varphi_2) = | \Proj(\varphi_1) \triangle \Proj(\varphi_2) |,
\end{equation*}
where the measure $|\cdot|$ is the Lebesgue measure.
\end{define}

\begin{theorem}[Soundness]
\label{th:dist-equivalence}
If $\phi_1$ and $\phi_2$ are two STL formulae with the same language,
then they cover the same space.
Formally, we have $\CA{L}(\phi_1) = \CA{L}(\phi_2)$ implies
$\Proj(\phi_1) = \Proj(\phi_2)$,
and thus $d_{\symdiff}(\phi_1, \phi_2) = 0$.
\end{theorem}
\begin{proof}
The proof is trivial since $\Proj$ is essentially a projection
of the language of an STL formula $\phi$ onto space-time
$\BB{S} \times [0, \norm{\phi})$.
\end{proof}

}

\begin{define}
\label{def:symm-diff-dist}
The \symdiff metric is defined as:
\begin{equation*}
d_{\symdiff}(\varphi_1, \varphi_2) = \frac{1}{\maxhorz+1} | \CA{L}(\varphi_1) \triangle \CA{L}(\varphi_2) |,
\end{equation*}
where $|\cdot|$ is the Lebesgue measure.
\end{define}

\begin{proposition}
The $(\Phi^{\CA{S}_T}/\equiv, d_{\symdiff})$ is a metric space.
\end{proposition}
\begin{proof}
Follows immediately from the definition.
The metric is well-defined since the formulae have
time horizons bounded by $\maxhorz$ over discrete-time signals.
Their languages are compact subsets of the Euclidean space
$\BB{R}^{n(\maxhorz+1)}$.
Thus, the Lebesgue measure is defined.
\end{proof}

We define the coverage of signal sets in the space-time value set. 
Formally, we have the map $\Proj: 2^{\CA{S}_T} \to \BB{S} \times T\BB{U} $ such that
$\Proj(S) = \bigcup_{s\in S} \bigcup_{t \in [0, T]} \{(s_t, \tau) \mid t \leq \tau \leq t+1\}$, where $S \subseteq \CA{S}$.
For an STL formula $\phi$, $\Proj(\phi) = \Proj(\CA{L}(\phi))$.

Let $\CA{S}=\BB{U}^n$, and $p = x_i \leq \mu$ be a rectangular predicate with
$\mu \in \BB{U}$ and $i\in \{1,\ldots, n\}$.
The coverage of $p$ is
$\Proj(p) = \left((\BB{U}^{i-1} \times \mu\BB{U} \times \BB{U}^{n-i}) \times \BB{U}\right) \cup \left(\BB{U}^n \times \{1 \leq t \leq \maxhorz\} \right)$.

\begin{theorem}
\label{th:proj-soundness}
If $\phi_1$ and $\phi_2$ are two STL formulae with the same language,
then they cover the same space.
Formally, we have $\CA{L}(\phi_1) = \CA{L}(\phi_2)$ implies
$\Proj(\phi_1) = \Proj(\phi_2)$.
\end{theorem}
\begin{proof}
Immediately follows from the fact that $\Proj$ is a projection of $\CA{L}(\phi)$ onto space-time
$\BB{S} \times T\BB{U}$.
\end{proof}

Note that the converse is not true in general.
In particular, it can fail for formulae containing disjunctions.

\journal{
\begin{define}
Let $(\Omega, \CA{F}, \lambda)$ be a  measure space,
the measurable space $(\BB{S}, \Sigma)$,
and $X$ be a random variable with index set $\BB{T}$.
The measure on the space of signals $\CA{S}$ is defined as:
\begin{eqnarray}
\lambda^\CA{S}(S) = \inf_{\Omega' \in \CA{F} \land X(\Omega', \BB{T}) \subseteq S} \lambda(\Omega').
\end{eqnarray}

\end{define}


\begin{assumption}
\label{assump:predicate-set}
The value set $\BB{S}$ is a polyhedron, and for every predicate $p(x)$ we have that $\{x \in \BB{S} \mid p(x)=\True\}$ admits a finite partition into convex polyhedra.
\end{assumption}

The assumption is equivalent to the case of predicates induced
by finite sets of affine support functions.
In other words, the functions $f$ underlying are piecewise affine and
have a finite number of pieces (cells).
{\color{orange}[Cristi: Sadra, do you know a name for such functions?]}

In the following, we will consider convex polyhedra with full, partial, or
no boundaries.

\begin{theorem}
\label{th:language-representation}
For any STL formula $\phi \not\equiv \False$ defined over the compact set $\BB{S}$
of signal values, we have that $\Proj(\phi)$ admits a finite partition of convex polyhedral subsets of $\BB{S} \times \left[ 0, \norm{\phi} \right)$.
\end{theorem}
\begin{proof}
The property follows by structural induction over the structure of STL formulae.
First note that all sets must be bounded since
$\Proj(\phi) \subseteq \BB{S} \times \left[ 0, \norm{\phi} \right)$ for all $\phi \in \Phi^\CA{S}$.
The base cases are $\Proj(\True) = \BB{S} \times [0, h)$ and
$\Proj(p) = \{(x, 0) \mid x \in \BB{S}, p(x)=\True\} \cup \left(\BB{S} \times (0, h)\right)$,
where $h \in \BB{R}_{\geq 0}$ is some time horizon.
The first case it trivially true, while the second follows from the assumption
that $\phi \not\equiv \False$ which implies that there is at least one
signal that satisfies the predicate at time 0, and Assump.~\ref{assump:predicate-set}.
Note that the intersection, union, and complementation with respect to $\BB{S} \times [0, h)$, $h\geq 0$, of convex polyhedra produce finite collections of convex polyhedra.
Thus, the induction step holds trivially for conjunction, and complementation.
The case for the until operator follows from the observation that $\Proj(\phi_1 \Until_{[a, b)} \phi_2)$ is the union of $\Proj(\phi_1) \times [0, b)$ and $\Proj(\phi_2) \times [a, b)$ projected onto $\BB{S} \times [0, h)$.
Both sets are polyhedra in $dim(\BB{S})+2$, and the projection operator preserve their polyhedral structure.
This concludes the proof since all induction cases hold.
\end{proof}

\begin{corollary}
If all predicates are rectangular, then the partition of $\Proj(\phi)$
is composed of hyper-boxes.
\end{corollary}
}

\subsection{Comparison}

While the \hausdorff distance and the \symdiff difference are both metrics, they have quite different behaviors. Here, we elaborate on these differences and show an illustrative example. 

Informally, the \hausdorff distance has a stronger \emph{spatial} notion, and it is closely connected to STL score, as stated in Corollary \ref{corollary_signed} and Corollary \ref{corollary_neighborhood}. The \hausdorff distance, captures the worst-case spatial difference between formulae. On the other hand, the \symdiff is a more \emph{temporal} notion, as the areas also capture the length of temporal operators. It is possible that two STL formulae have a large \hausdorff distance, but a small \symdiff distance, and vice versa.
In applications, the choice is dependent on the user. The most useful may be a convex combination - which is a metric by itself - with a user-given convex coefficient. 
  
\begin{example}
\label{example_table}
Consider the six STL formulae in \eqref{eq_formulas}. We compute the \hausdorff and the \symdiff distances between all pairs of formulae using the methods proposed in Section~\ref{sec:methods}. The directed \hausdorff distances are also reported. 

1) Directed \hausdorff distance: The results are shown in Table~\ref{table_PH}, where the value in the $i$'th row and $j$'th column is $\vec{d}_{\hausdorff}(\phi_i,\phi_j)$. Each maximizer (the \hausdorff distance) is bolded. 

We have also included the truth constant in the distance table. It is observed that $\vec{d}_{\hausdorff}(\phi_i,\True)=0, \forall i \in \{1,\cdots,6\}$, which implies the fact that the language of each formula is contained within the language of $\True$, which is the set of all signals. The opposite direction, $\vec{d}_{\hausdorff}(\True,\phi)={d}_{\hausdorff}(\phi,\True)$, is, informally, the quantification of how \emph{restrictive} $\phi$ is.

Note that $d_{PH}(\phi_1,\phi_2)=0.04$. It is observed that most values are either $0.6=\max(1-0.4,0.2-0)$ or $0.56=\max(1-0.44,0.2-0)$, which correspond to the extreme signal that one language  contains but the other does not, or it is $0$, indicating that one language is a subset of another. For instance, $\phi_3$ is a ``weak'' specification in the sense that its language is \emph{broad} - any signal with some value in $\theta_1$ at some time satisfies it - so the directed distances from other formulae to $\phi_3$ are zero. Another notable example is the relation between $\phi_1$ and $\phi_4$. The directed PH distance is zero in both directions - the two formulae are equivalent. This is due to the fact that any signal that satisfies $\phi_4$, already satisfies $\phi_1$. The other direction also trivially holds. Note that some pairs, like $\phi_2$ and $\phi_3$, have non-zero \hausdorff distances in both directions. 

2) \symdiff distance: The results are shown in Table~\ref{table_PHSD} along with the \hausdorff distances for comparison.
Here, it can be seen that in most cases, the \symdiff distance is either a lot larger or a lot smaller than the \hausdorff distance.
This is largely due to the fact that this metric is based on area which is particularly highlighted when comparing any of the formulae to $\True$.
Since each formula's satisfaction space is very small in comparison to the entire bounded signal space, each of these values is quite large.
In contrast, the \symdiff distance between $\phi_1$ and $\phi_5$ is fairly small since the satisfaction regions for each of these formulae cover a similar area.
The \symdiff distance between $\phi_2$ and $\phi_3$ is on the larger side as their areas of satisfaction are quite different; however, these areas are still much closer to each other than they are to the entire bounded satisfaction area represented by $\True$.
Similar to the \hausdorff distance, the \symdiff distance between $\phi_1$ and $\phi_4$ is zero as they have completely overlapping areas of satisfaction.

The results in Table~\ref{table_PHSD} illustrate that there are different situations when the \hausdorff distance might be favored over the \symdiff distance and vice versa.
In cases where one cares about the area covered by the satisfaction region of a formula, the \symdiff distance should be used.
For instance, the \symdiff could be used to find a formula close to one that requires a signal to be held at a particular value for a long time interval.
However, if one only cares about how close the signal bounds of the formulae are to each other, the \hausdorff distance should be used.


\begin{table}
\centering
\vspace{2mm}
\caption{Example \ref{example_table}: Directed PH Distances}
\begin{tabular}{ | c | c | c | c | c | c | c | c | }
\hline
$\vec{d}_{PH}$ & $\True$ & $\phi_1$ & $\phi_2$ & $\phi_3$ & $\phi_4$ & $\phi_5$ & $\phi_6$ \\ \hline
$\True$ & \cellcolor{black!40}$0$ & $\bf 0.6$ & $\bf 0.56$ & $\bf 0.6$ & $\bf 0.6$ & $\bf 0.6$ & $\bf 0.6$ \\ \hline
$\phi_1$ & $0$ & \cellcolor{black!40}$0$  & $0$ & $0$ & $\bf 0$ & $0$ & $0$ \\ \hline
$\phi_2$ & $0$ & $\bf 0.04$ & \cellcolor{black!40}$0$ & ${0.04}$ & $\bf 0.04$ & $0.04$ & $0.04$  \\ \hline
$\phi_3$ & $0$ & $\bf 0.6$ & $\bf 0.56$ &  \cellcolor{black!40}$0$ & $\bf 0.6$ & $\bf 0.6$ & $\bf 0.6$  \\ \hline
$\phi_4$ & $0$ & $\bf 0$ & $0$ & $0$ & \cellcolor{black!40}$0$ & $0$ & $0$ \\ \hline
$\phi_5$ & $0$ & $\bf 0.6$ & $\bf 0.56$ & $0$ & $\bf 0.6$ & \cellcolor{black!40}$0$ & $0.04$ \\ \hline
$\phi_6$ & $0$ & $\bf 0.6$ & $\bf 0.56$ & $0$ & $\bf 0.6$ & $\bf 0.6$ & \cellcolor{black!40}$0$ \\ \hline
\end{tabular}
\label{table_PH}
\end{table}

\begin{table}
\centering
\caption{Example \ref{example_table}: \hausdorff and \symdiff Distances}
\begin{tabular}{ | c | c | c | c | c | c | c | c | }
\hline
\backslashbox{$\!\!\!d_{PH}\!\!\!$}{$\!\!\!d_{SD}\!\!\!$} & $\True$ & $\phi_1$ & $\phi_2$ & $\phi_3$ & $\phi_4$ & $\phi_5$ & $\phi_6$ \\ \hline
$\True$ & \cellcolor{black!40}$0$ & $0.8$ & $0.76$ & $0.99$ & $0.8$ & $\!0.804\!$ & $0.84$ \\ \hline
$\phi_1$ & $0.6$ & \cellcolor{black!40}$0$  & $0.04$ & $0.19$ & $0$ & $\!0.036\!$ & $0.03$ \\ \hline
$\phi_2$ & $0.56$ & $0.04$ & \cellcolor{black!40}$0$ & $0.23$ & $0.04$ & $\!0.044\!$ & $0.07$  \\ \hline
$\phi_3$ & $0.6$ & $0.56$ & $0.56$ &  \cellcolor{black!40}$0$ & $0.19$ & $\!0.186\!$ & $0.16$  \\ \hline
$\phi_4$ & $0.6$ & $0$ & $0.04$ & $0.6$ & \cellcolor{black!40}$0$ & $\!0.036\!$ & $0.03$ \\ \hline
$\phi_5$ & $0.6$ & $0.56$ & $0.56$ & $0.6$ & $0.6$ & \cellcolor{black!40}$0$ & $\!0.066\!$ \\ \hline
$\phi_6$ & $0.6$ & $0.56$ & $0.56$ & $0.6$ & $0.6$ & $0.6$ & \cellcolor{black!40}$0$ \\ \hline
\end{tabular}
\label{table_PHSD}
\end{table}

\end{example}

%% file: methods.tex
\section{Computation}
\label{sec:methods}

This section presents algorithms for computing the \hausdorff and the \symdiff distances between STL specifications.

\subsection{Pompeiu-Hausdorff Distance} 
In this section, we propose an optimization-based method to compute the \hausdorff distance between two STL formulae. 
\begin{define}
Given an STL formula $\varphi$ that contains no negation, we define $\varphi^{\epsilon+}$ with the same logical structure as $\varphi$ with predicates replaced as follows:
\begin{itemize}
\item $f(x) \ge \mu$ replaced with $f(x) \ge \mu - \epsilon$;
\item $f(x) \le \mu$ replaced with $f(x) \le \mu + \epsilon$.
\end{itemize} 
\end{define}
Intuitively, $\varphi^{\epsilon+}$ is a relaxed version of $\varphi$. It is easy to verify from \eqref{eq:quantitative-semantics} that if $\rho(s,\varphi^{\epsilon+},0)=\rho(s,\varphi,0)+\epsilon, \forall s \in \CA{S}_T$. 
\begin{lemma}
The following relation holds:
\begin{equation}
\label{eq_PH_min}
\vec{d}_{PH}=\min \{\epsilon \ge 0 \mid \CA{L}(\varphi_1) \subseteq \CA{L}(\varphi_2^{\epsilon+}) \}
\end{equation}
\end{lemma}
\begin{proof}
(sketch)
The result is a direct consequence of \eqref{PH_ball} as we have $\CA{L}(\varphi_2^{\epsilon+})=\CA{L}(\varphi_2)+\epsilon \CA{B}^{\CA{S}_T}$.  
\end{proof}
The following statement provides the main result, and the base for the computational method of this section. 

\begin{theorem}
\label{thm_PH}
Given $\varphi_1,\varphi_2 \in \Phi^{\CA{S}_T}$, $\CA{L}(\varphi_1) \neq \emptyset$, define $\epsilon^*$ as the following optimum:
\begin{equation}
\label{eq_one_way_MILP}
\begin{array}{lll}
\epsilon^* = & \max & \epsilon, \\
& \text{subject to} & s \models \varphi_1, s \not\models \varphi_2^{\epsilon+}, \\ 
&& \epsilon \ge 0, s \in \mathcal{S}_T.
\end{array}
\end{equation}
Then the following holds:
\begin{equation}
\vec{d}_{PH}(\varphi_1,\varphi_2)= \left\{ 
\begin{array}{ll}
\epsilon^* & \eqref{eq_one_way_MILP} \text{ is feasible}, \\
0 & \text{otherwise}.
\end{array}
\right.
\end{equation}
\end{theorem}
\begin{proof}
First, consider the case that \eqref{eq_one_way_MILP} is infeasible or its value is $0$. Then, it means that the constraints $s \in \CA{L}(\varphi_1), s \not \in \CA{L}(\varphi_2^{\epsilon+})$ are infeasible for all $\epsilon > 0$, which implies that for any $s \in \CA{L}(\varphi_1)$, we have $\rho(s, \CA{L}(\varphi_2),0) \ge 0$. Thus, $\CA{L}(\varphi_1) \subset \CA{L}(\varphi_2)$ and consequently $\vec{d}_{PH}(\varphi_1,\varphi_2)=0$.

Now consider the case \eqref{eq_one_way_MILP} is feasible and $\epsilon^*>0$.  Then, $s \not \models \varphi_2^{\epsilon^*+}$ is an active constraint, which implies $\rho(s,\varphi_2^{\epsilon^*+},0)=0$. Note that $s$ is also optimized in \eqref{eq_one_way_MILP}. Thus $\rho(s,\varphi_2,0)<0$, or $s \not \in \CA{L}(\varphi_2)$. 
We can rewrite \eqref{eq_PH_min} as:
\begin{equation}
\label{eq_PH_max}
\vec{d}_{PH}=\sup \{\epsilon \ge 0 \mid \CA{L}(\varphi_1) \not \subseteq \CA{L}(\varphi_2^{\epsilon+}) \}.
\end{equation}
Note that we have used $\sup$ instead of $\max$ as $\CA{L}(\varphi_1) \not \subseteq \CA{L}(\varphi_2^{\epsilon+})$ is a strict relation. Also note that such the supremum exists as i) the condition is satisfied for $\epsilon=0$ and ii) the language sets are bounded. We show that \eqref{eq_one_way_MILP} captures \eqref{eq_PH_max}. If $\CA{L}(\varphi_1) \not \subseteq \CA{L}(\varphi_2^{\epsilon+})$, then it means that $\exists s \models \varphi_1$ but $\rho(s,\varphi_2^{\epsilon^+},0)<0$. This is what is captured by the constant in $\eqref{eq_one_way_MILP}$, with the difference that $\rho(s,\varphi_2^{\epsilon^+},0)<0$ is replaced by a non-strict inequality and $\sup$ is replaced by $\max$.  
\end{proof}
We convert \eqref{eq_one_way_MILP} into a MILP problem. The procedure for converting STL into MILP constraints is straightforward, see, e.g., \cite{raman2014model}. The encoding details are omitted here. By solving two MILPs, we are able to obtain the \hausdorff distance. Two MILPs can be aggregated into a single MILP, but that usually more than doubles the computation time due to larger branch and bound trees. Moreover, it is often useful to have the knowledge of the directed \hausdorff distances. 

Theorem \ref{thm_PH} requires that formulae do not contain negation. Negation elimination is straightforward: first, the formula is brought into its Negation Normal Form (NNF), where all negations appear before the predicates. Next, the predicates are \emph{negated}. For example, we replace $\neg (x \le \mu)$ by $(x \ge \mu)$. We remind the reader that we do not consider strict inequalities, hence $\neg (x \le \mu)$ and $(x \le \mu)$ are both true if $x_0=\mu$. 
Finally, observe that the choice of $T$ does not effect the values of \hausdorff distance, as long as it is larger than the horizons of two formulae that are compared. Given $\phi_1, \phi_2 \in \Phi^{\CA{S}_T}$, the values of $s_t$ for $t>\max\{\norm{\phi_1},\norm{\phi_2}\}$ do not have any associated constraints in \eqref{eq_one_way_MILP}. 
  
\subsection*{Complexity}
The complexity of \eqref{eq_one_way_MILP} is exponential in the number of integers, which grows with the number of predicates and horizons of the formulae. However, since signal values do not have any dynamical constraints, we found solving \eqref{eq_one_way_MILP} to be orders of magnitudes faster than comparable STL control problems, such as those studied in \cite{raman2014model}. All the values obtained in Table \ref{table_PH}  were evaluated almost instantaneously using Gurobi MILP solver on a personal computer. 

\begin{figure}[tb]
\centering
\subfloat[$\phi_1$]{\includegraphics[width=0.30\linewidth]{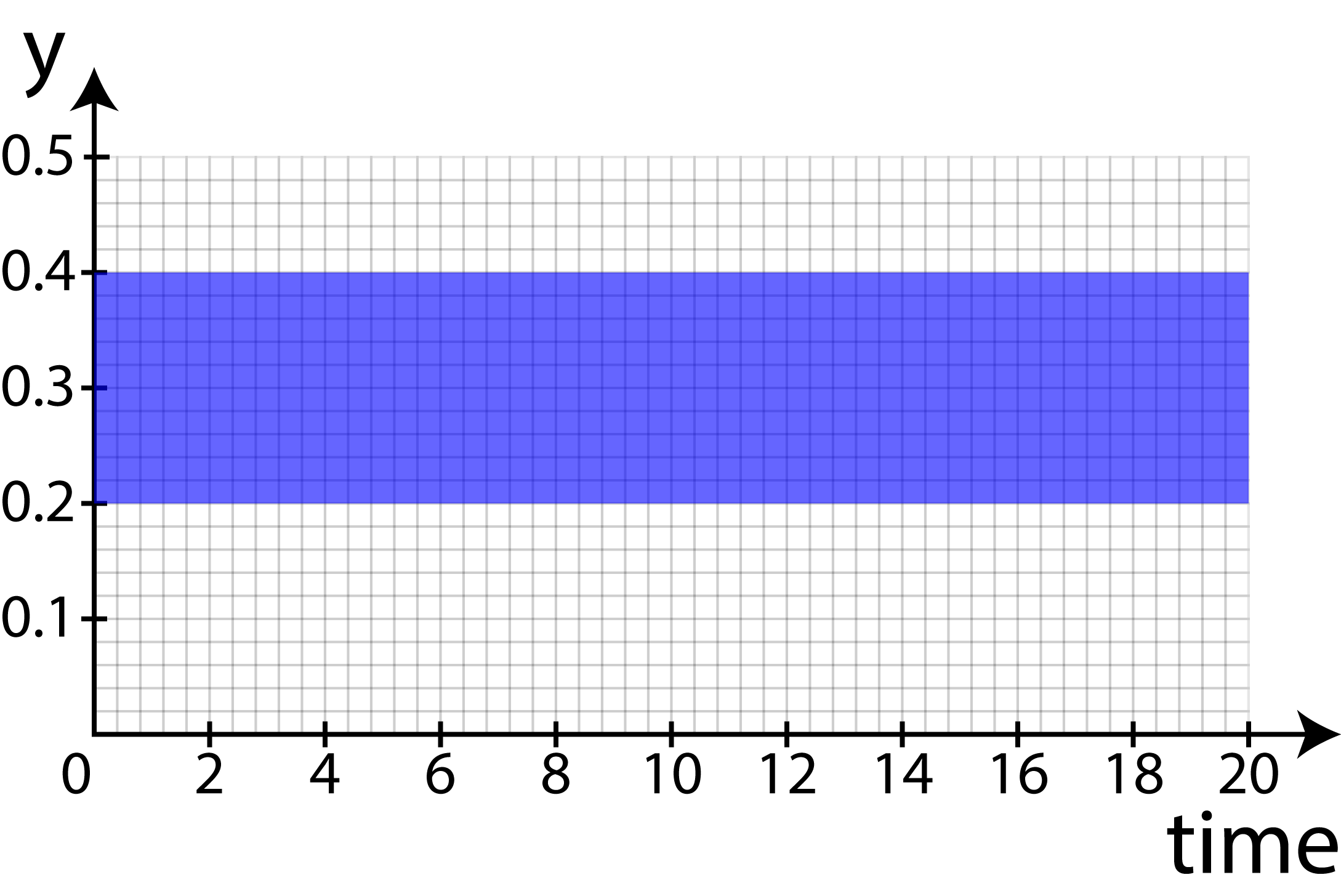}
\label{fig:phi1}}
\subfloat[$\phi_5$]{\includegraphics[width=0.30\linewidth]{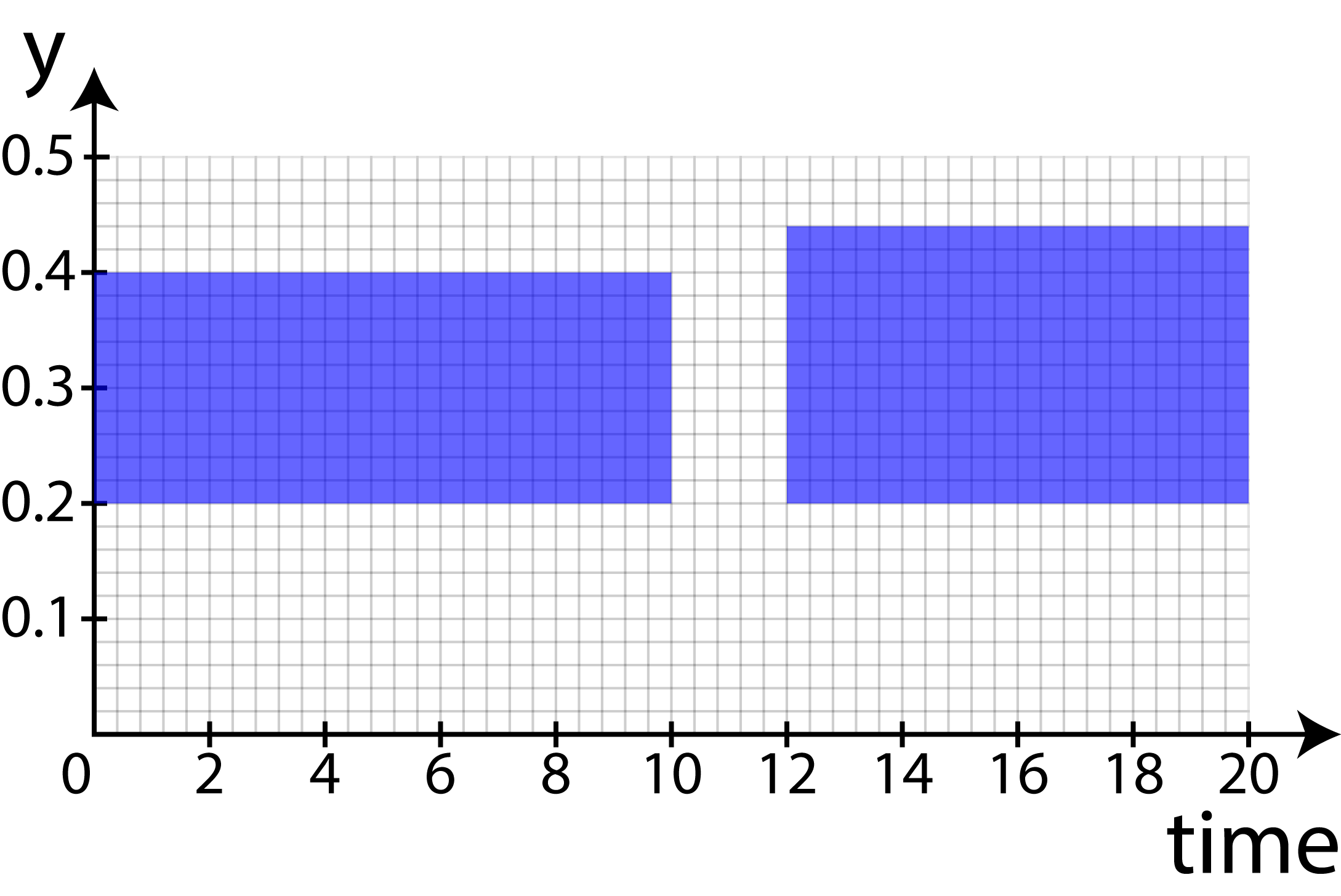}
\label{fig:phi5}}
\subfloat[Overlap $\phi_1$ and $\phi_5$]{\includegraphics[width=0.30\linewidth]{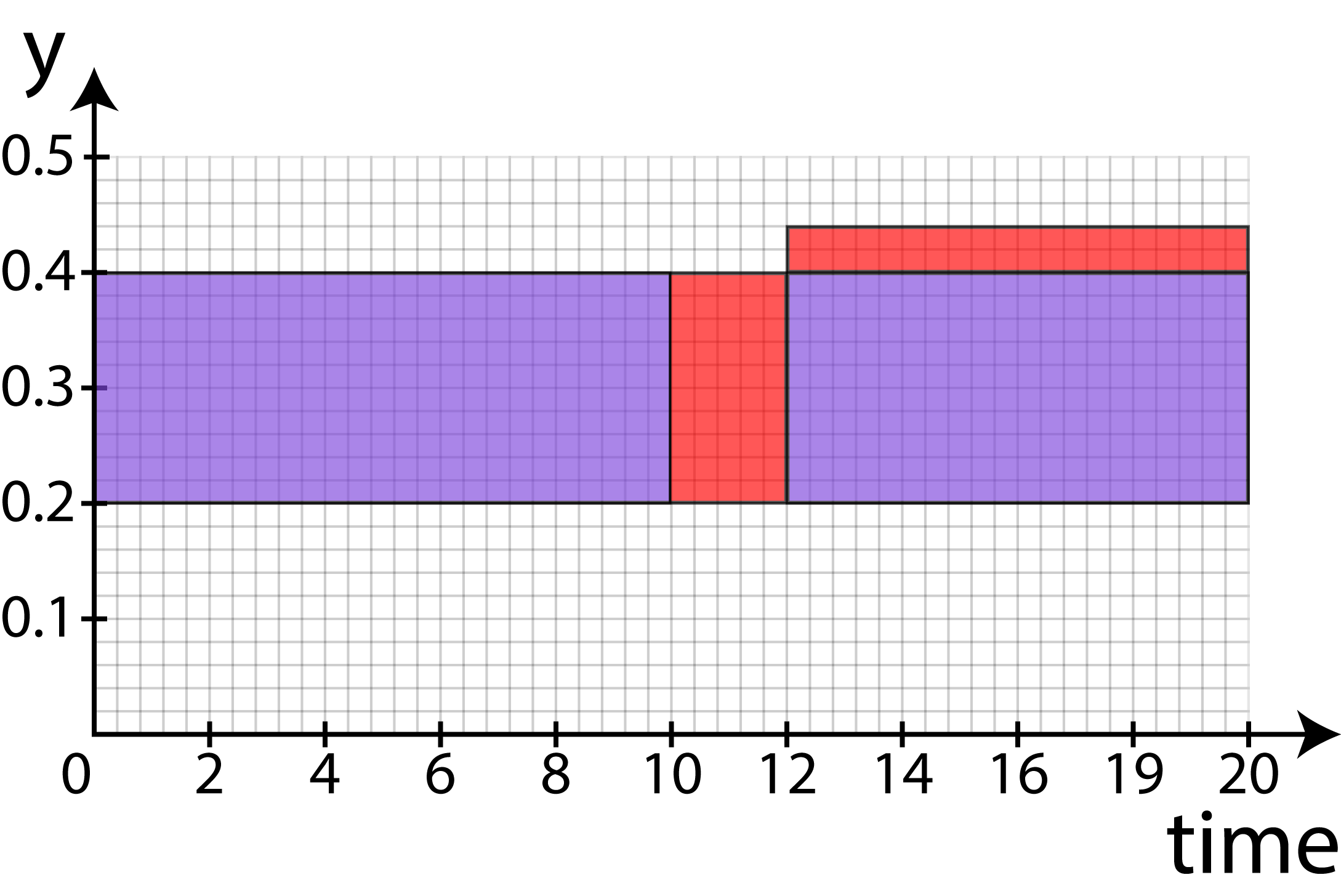}
\label{fig:overlap}}
\caption{\small \protect\subref{fig:phi1} and \protect\subref{fig:phi5} show the area of satisfaction boxes for $\phi_1$ and $\phi_5$ from Example~\ref{example_table}, respectively.
The blue regions represent the boxes that are computed for globally ($\Always$) operators.
In \protect\subref{fig:overlap}, the red regions represent the non-overlapping area and the purple regions represent the overlapping area between $\phi_1$ and $\phi_5$.
The \symdiff distance for this example is the area of the red regions ($(2 \times 0.2) + (8 \times 0.04) = 0.72$) divided by the maximum time horizon which is $\frac{0.72}{20} = 0.036$.}
\label{fig:example1boxes}
\end{figure}

\subsection{Symmetric Difference}

This section presents an algorithm for computing boxes representing the area of satisfaction of a formula as well as a method for determining the \symdiff between two sets of boxes.
Each set of boxes approximates the projection ($\Proj$) of the formula and represents the valid value-space that a time-varying signal can take such that traces that are contained entirely within the boxes satisfy the formula.

Computing the set of boxes representing the area of satisfaction is a recursive process that takes as input an STL formula, $\phi$, a set of max values, $X_{max}$, for each signal, $x \in X$ (used to normalize the signal values to a unit space), and a discretization threshold, $\delta$.
This algorithm, $AoS$, is presented in Algorithm~\ref{alg:boxes}
Here, $box(t_1,t_2,x_1,x_2, i) = \Proj(\Always_{[t_1, t_2]} x_1 \leq x^i \leq x_2)
\subseteq \BB{S}\times T\BB{U}$
creates a new box with minimum and maximum times $t_1$ and $t_2$, minimum and maximum values $x_1$ and $x_2$, and spatial dimension $i \in \{1, \ldots, n\}$, respectively;
$overlap$ determines the time window of the overlap between two boxes;
$combine$ 
takes two boxes and produces a set of boxes representing the intersection of the overlapping time window region;
$b.lt$, $b.ut$, $b.lv \in \BB{S}$, and $b.uv \in \BB{S}$ return the lower time window, upper time window, lower variable, and upper variable values for box $b$, respectively;
$*$ in the definition of a box denotes that it may restrict multiple spatial dimensions;
and the $\parallel$ operator is used to create a ``choice'' set representing that either of the two sets separated by it can be selected as the set of boxes representing the area of satisfaction.

\SetAlFnt{\footnotesize}
\begin{algorithm}[!htb]
\DontPrintSemicolon
\SetKwInOut{Input}{input}\SetKwInOut{Output}{output}
\Input{STL formula $\phi$, max value set $X_{max}$, discretization threshold $\delta$.}
\Output{Set of boxes $\mathcal{B}$ for each signal in $\phi$.}
\BlankLine
\lIf{$\phi := x^i \leq \pi$}{
  \Return{$box(0,0,0,\pi/x_{max}, i)$}
}
\lElseIf{$\phi := x^i > \pi$}{
  \Return{$box(0,0,\pi/x_{max},1, i)$}
}
\uElseIf{$\phi := \phi_1 \land \phi_2$}{ 
  Create a new set $\mathcal{B}$\;
  \For{\textbf{\textup{each}} $b_1 \in AoS(\phi_1)$ \textup{and} \textbf{\textup{each}} $b_2 \in AoS(\phi_2)$}{
    \lIf{$overlap(b_1,b_2)$}{
      Add $combine(b_1,b_2)$ to $\mathcal{B}$
    }
    \lElse{
      Add $b_1$ to $\mathcal{B}$ and $b_2$ to $\mathcal{B}$
    }
  }
  \Return{$\mathcal{B}$}\;
}
\lElseIf{$\phi := \phi_1 \lor \phi_2$}{
  \Return{$AoS(\phi_1) \parallel AoS(\phi_2)$}
}
\uElseIf{$\phi := \Always_{[t_1,t_2]}(\phi_1)$}{
  Create a new set $\mathcal{B}$\;
  \For{\textbf{\textup{each}} $b \in AoS(\phi_1)$}{
    $b' = box(b.lt+t_1, b.ut+t_2, b.lv, b.uv, *)$\;
    Add $b'$ to $\mathcal{B}$\;
  }
  \Return{$\mathcal{B}$}\;
}
\ElseIf{$\phi := \Event_{[t_1,t_2]}(\phi_1)$}{
  \Return{$AoS\left(\displaystyle\bigvee_{i=1}^{(t_2 - t_1)/\delta}\Always_{[t_1 + \delta (i-1),t_1 + \delta i]}(\phi_1)\right)$}\;
}
\caption{Convert to Area of Satisfaction Boxes ($AoS$)\label{boxes}}
\label{alg:boxes}
\end{algorithm}


To address the problem of projection of formulae containing disjunction (the converse to  Theorem~\ref{th:proj-soundness}), $AoS$ utilizes the $\parallel$ operator.
If this algorithm instead generated boxes representing the projection of all formulae, it would be possible for the satisfaction space represented by the boxes to capture signals that the original formula does not allow.
The application in Section~\ref{sec:design-quality} highlights this problem and presents a way of dealing with it for that particular example.

For operators such as globally ($\Always$), $AoS$ is exact and produces boxes bound by the time bounds of the operator that represent the projection of the primitive.
However, operators such as eventually ($\Event$) do not immediately lend themselves to conversion into a set of boxes.
In order to deal with this operator, we approximate it by converting it into a disjunction of globally predicates.
Each globally predicate is generated using a small threshold value ($\delta$) for its time window width.
The new formula requires that the expression be true in at least one of the smaller time windows essentially introducing a mandatory $\delta$ ``hold'' time for eventually operators.
The tunability of $\delta$ allows for a user to give up some accuracy for gains in performance of box computation and ultimately \mudist comparison.
Examples of computing the area of satisfaction boxes for some of the formulae in Example~\ref{example_table} are shown in Figure~\ref{fig:example1boxes}.

\journal{
In the limit as $\delta$ approaches zero, each of the globally predicates will enforce that the expression is satisfied in a single time instance.
A proof of this equality in semantics is as follows:

\begin{proof}
{\footnotesize\begin{align*}
& s \models (\Event_{(t_1,t_2)} \phi) &\Leftrightarrow \ & \lim_{\delta \to 0} \big(s \models (\Always_{(t_1,t_1 + \delta)} \phi) \: \lor\\
&&& s \models (\Always_{(t_1 + \delta,t_1 + 2\delta)} \phi) \\
&&&\lor ... \lor\\
&&& s \models (\Always_{(t_2 - 2\delta,t_2 - \delta)} \phi)\\
&&& \lor s \models (\Always_{(t_2 - \delta,t_2)} \phi)\big)\\
& s \models (\True \Until_{(t_1,t_2)} \phi) &\Leftrightarrow \ & \lim_{\delta \to 0}  \big(\forall t_u \in (t_1,t_1 + \delta) \ s[t_u] \models \phi \: \lor\\
&&&\forall t_u \in (t_1+ \delta,t_1 + 2\delta) \ s[t_u] \models \phi \\
&&&\lor ... \lor\\
&&&\forall t_u \in (t_2 - 2\delta,t_2 - \delta) \ s[t_u] \models \phi \\
&&&\lor \forall t_u \in (t_2 - \delta,t_2) \ s[t_u] \models \phi\big)\\
& \exists t_u \in (t_1,t_2) \text{ s.t. } \big(s[t_u] \models \phi\big) \!\!\!\!\!\! &\Leftrightarrow \ & \exists t_u \in (t_1,t_2) \text{ s.t. } \big(s[t_u] \models \phi\big)
\end{align*}}
\end{proof}
}


The \symdiff between two sets of boxes is computed by calculating the area of the sum of the non-intersected area for each box set.
This value is normalized by the maximum time horizon, $\maxhorz$, and results in the \symdiff computation:
\begin{equation*}
\label{eq:area_formula}
d_{\symdiff}(\CA{B}_{\phi_1}, \CA{B}_{\phi_2}) = \frac{1}{\maxhorz+1}\left|\left(\bigcup_{b_1\in\CA{B}_{\phi_1}} b_1 \right) \triangle \left(\bigcup_{b_2 \in\CA{B}_{\phi_2}} b_2 \right)\right|
\end{equation*}
Figure~\ref{fig:overlap} visually illustrates how the \symdiff between $\phi_1$ and $\phi_5$ is computed.

Note that the \symdiff distance is scaled by $\frac{T+1}{T'+T+1}$ if the maximum horizon
is increased by $T'$. This again shows the temporal nature of the \symdiff as opposed to
the \hausdorff distance which does not change.

\subsection*{Complexity}
The complexity of Algorithm~\ref{alg:boxes} depends on the complexity of the $\parallel$ operation which may be exponential depending on how it is implemented.
Otherwise, the algorithm is polynomial due to the box combination operations carried out whenever a conjunction predicate is encountered.
In practice, computing the \symdiff distance using this method for formulae with a few dozen
predicates typically takes only a few seconds.



%% file: case_study.tex


\section{Quantification of design quality}
\label{sec:design-quality}

In our first application, we show an example of how the proposed metrics can be used in behavioral synthesis.
Behavioral synthesis is an important process in design automation where the description of a desired behavior is interpreted and a system is created that implements the desired behavior.
Our goal is to check if the characterized implementations satisfy the specifications of a system.
Implementations include simulations and execution traces of a system. 
These implementations are characterized into formal specifications using TLI.
We show that the proposed metrics can be used in the synthesis step to choose a design from the solution space that can best implement the desired specification.
The specific example we have chosen to highlight this application is the synthesis of genetic circuits in synthetic biology. 

\subsection*{Synthetic Genetic Circuit Synthesis}

\journal{
Biological cells can produce a range of behaviors in response to the complex biochemical factors that influence the cell. 
One of the fundamental goals in synthetic biology is to program living cells with synthetic decision-making genetic circuits to produce a desired behavior in these biochemical environments~\cite{vaidyanathan2015framework}.
}
In this example, we have a set of desired behaviors (each formally represented by STL) which describe the various behaviors expected of a genetic circuit. 
This set of behaviors is referred to as a performance specification: $\set^\phi$. $\set^\phi$ consists of 2 \stl\ formulae: $\phi_{low}$ and $\phi_{high}$ which describe the desired amount of output produced by the genetic circuit over time:
\begin{align*}
& \phi_{low} &=& \:\:\Always_{[0,300]}(x<40) \land \Always_{[0,300]}(x>0) \\
& \phi_{high} &=& \:\:\Always_{[0,125]}(x<200) \land \Always_{[125,300]}(x<320)\:\:\land \\
&&& \:\:\Always_{[0,150]}(x>0) \land \Always_{[150,200]}(x>100)\:\:\land \\
&&& \:\:\Always_{[200,300]}(x>150) 
\end{align*}
In this case, the output of the circuit corresponds to the expression of a fluorescent protein. 
$\phi_{low}$ specifies that the output must consistently be below 40 units from time 0 to 300, and 
$\phi_{high}$ specifies that that output must gradually increase over time and must end up between 150 and 320 units between time 200 and 300.
Our solution space consists of two genetic circuits. The first circuit has a constitutive promoter as shown in Figure~\ref{fig:constitutiveSBOL}. 
Constitutive expression removes flexibility for consistency allowing constant protein production independent of the state or inputs of the system, which is highlighted in Figure~\ref{fig:constitutiveExperiment}. 
The second circuit has an inducible promoter: a sugar detecting transcription factor AraC*, which will turn on the protein production if and only if it is in the presence of a specific input molecule (arabinose) as shown in Figure~\ref{fig:inductionSBOL}. Figure~\ref{fig:inductionExperiment} shows the output of the circuit for various concentrations of arabinose. Both of these synthetic genetic circuits were built in \textit{Escherichia coli}. The traces were obtained from biological experiments by measuring fluorescence. 

\begin{figure}[tb]
\centering
\subfloat[Constitutive Expression]{\includegraphics[width=0.4\linewidth]{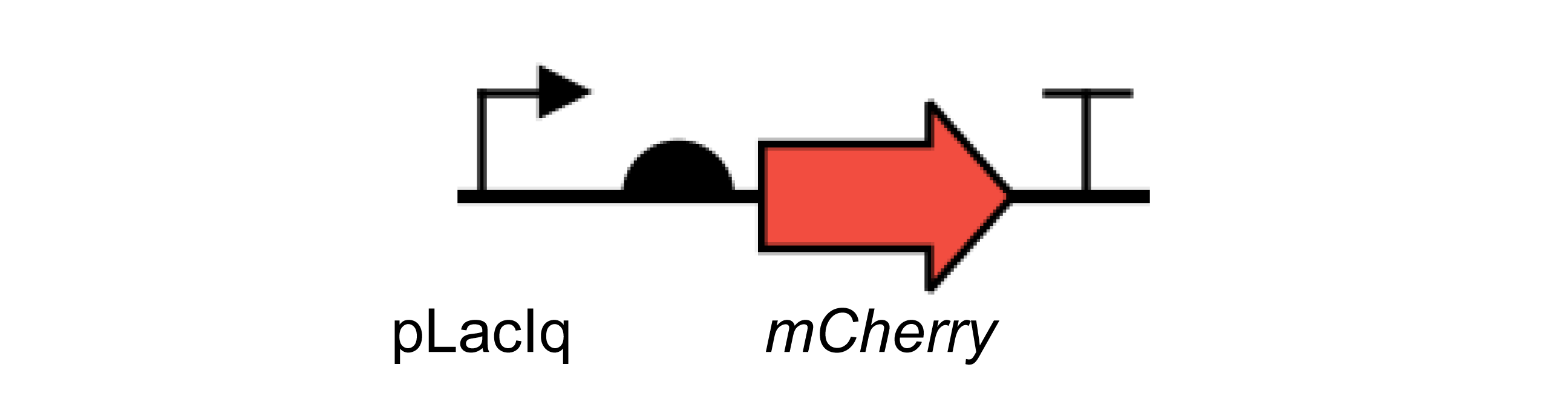}
\label{fig:constitutiveSBOL}}
\subfloat[Induction Circuit]{\includegraphics[width=0.4\linewidth]{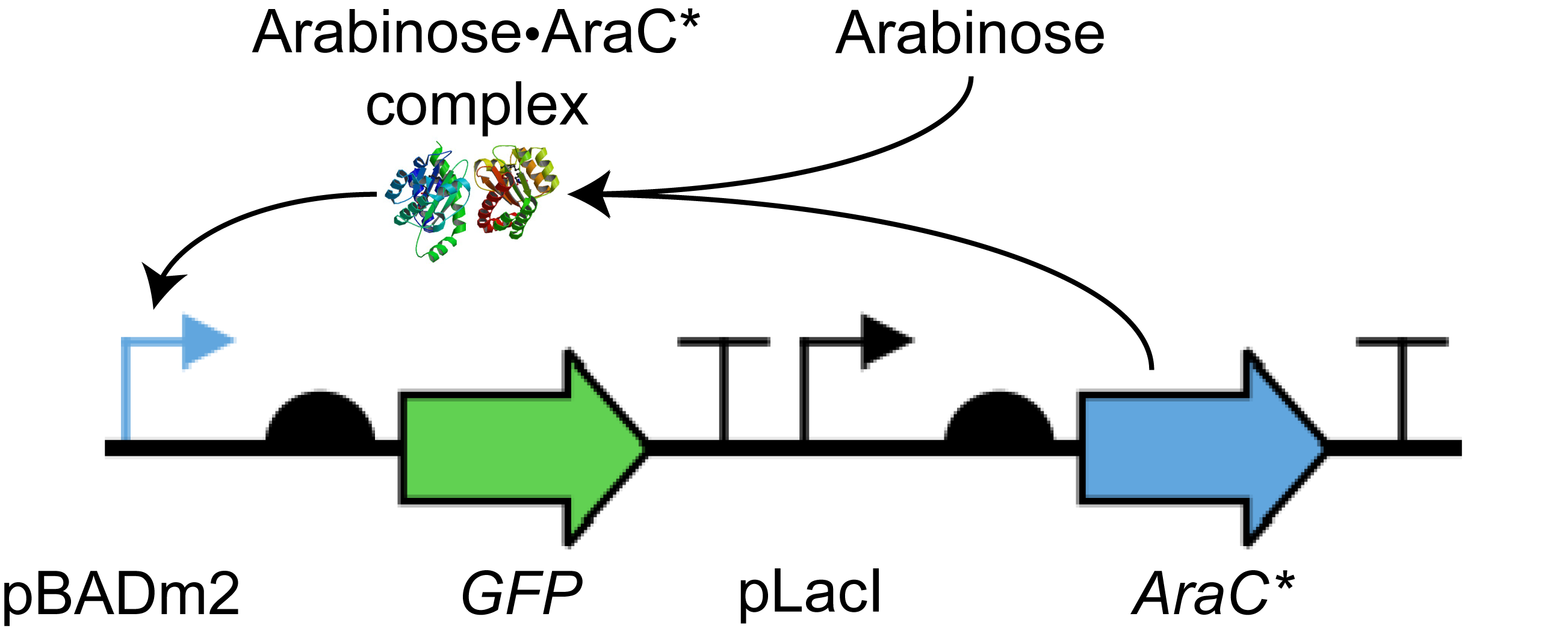}
\label{fig:inductionSBOL}}\\
\subfloat[Output of Constitutive Expression]{\includegraphics[width=0.4\linewidth]{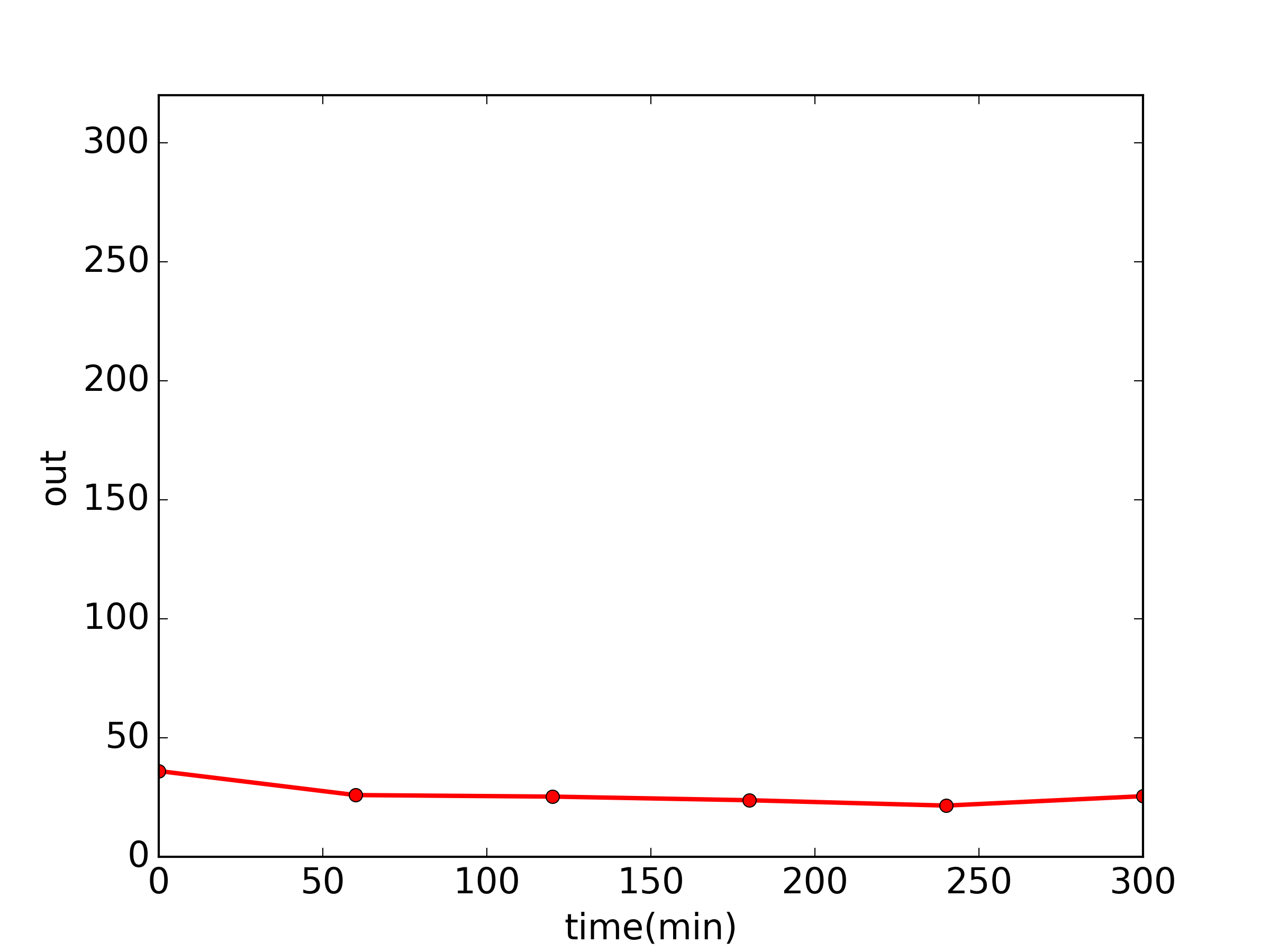}
\label{fig:constitutiveExperiment}}
\subfloat[Output of Induction Circuit]{\includegraphics[width=0.4\linewidth]{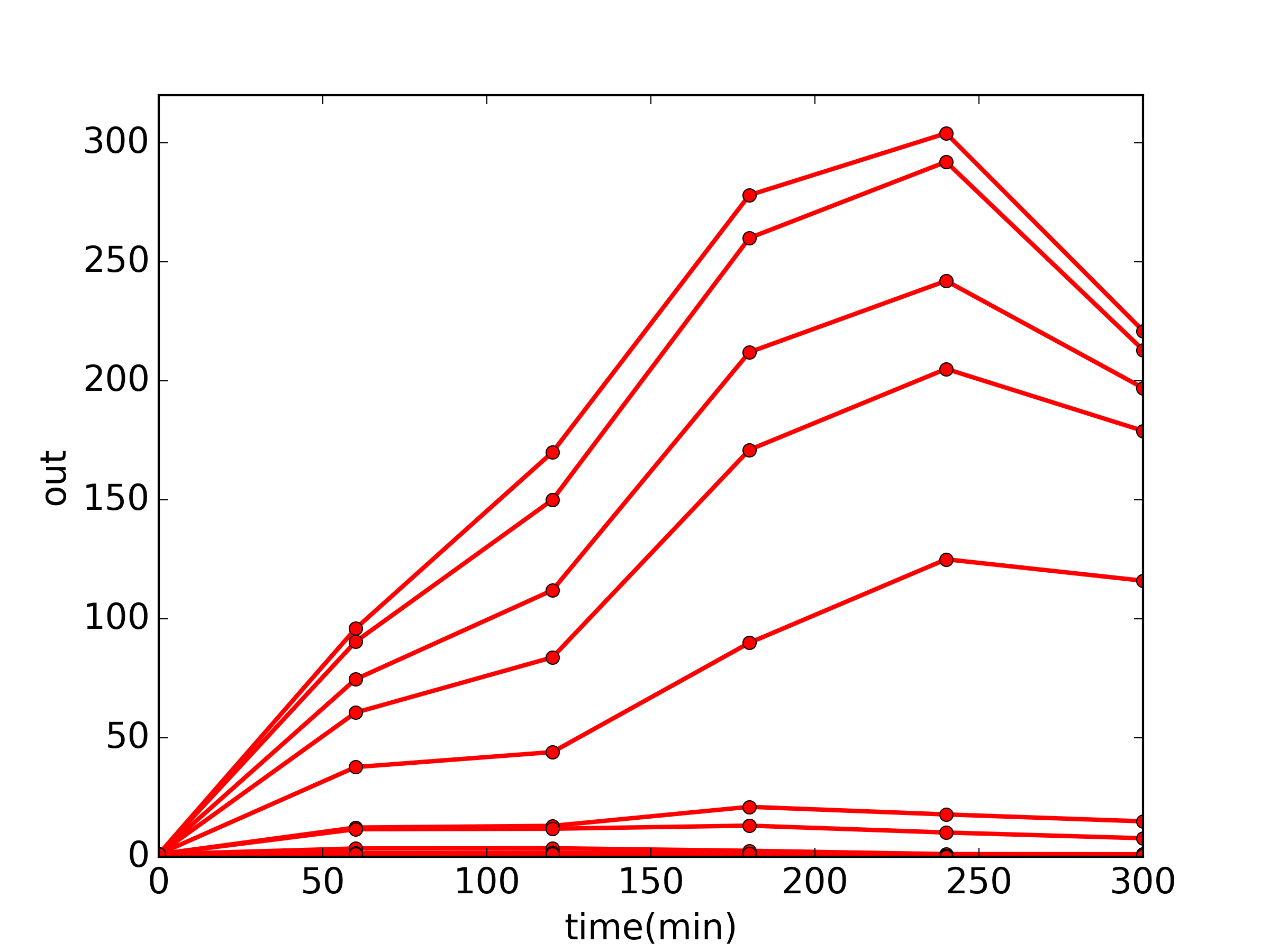}
\label{fig:inductionExperiment}}
\caption{\small \protect\subref{fig:constitutiveSBOL} and \protect\subref{fig:inductionSBOL} show SBOL Visual representations of the genetic circuits with a constitutive promoter and an inducible promoter, respectively. Biological traces in \protect\subref{fig:constitutiveExperiment} and \protect\subref{fig:inductionExperiment} were obtained by evaluating geometric mean fluorescence at regular intervals by flow cytometry.}
\label{fig:sbolCircuits}
\end{figure}


Our goal is to choose the circuit that can ``satisfy as many behaviors as possible'' in $\set^\phi$. 
It is important to note here that it is difficult to express the term ``satisfy as many behaviors as possible'' using the syntax and semantics of \stl. 
For instance, expressing the desired specification as a disjunction of all the formulae in $\set^\phi$ would imply that satisfying any one specification is sufficient for the genetic circuit to satisfy the performance specification.
Similarly, expressing the desired specification as a conjunction of all the formulae in $\set^\phi$ would imply that at any point in time, the output of a genetic circuit must have multiple distinct values, which is physically impossible. 

This conundrum is highlighted in the current example.
The output of constitutive expression satisfies $\phi_{low}$ but cannot satisfy $\phi_{high}$.
The induction circuit produces traces that can satisfy both $\phi_{low}$ and $\phi_{high}$.
However, traditional model checking techniques may not help a designer choose the desired circuit.
Using statistical model checking, for example, the circuit with constitutive expression yields a satisfaction likelihood of $1.0$ and the induction circuit yields a satisfaction likelihood of $0.83$ when checked against $\phi_{low} \lor \phi_{high}$.
With these results, one might think that the circuit with constitutive expression best satisfies the performance specification.
\journal {
whereas the output of the induction circuit yields a satisfaction likelihood of $0.5$ when checked against $\phi_{low}$ and $0.33$ when checked against $\phi_{high}$. It is also interesting to note that the satisfaction likelihood of the induction circuit when checked against $\phi_{low}\orltl\phi_{low}$ yields a value of $0.83$.
}

To address the issue of satisfying as many behaviors as possible, we treat the performance specification's region of satisfaction as the union of the regions of satisfaction of all the formulae in $\set^\phi$ as shown in Figure~\ref{fig:union}.
We compute this region by taking the union of the generated boxes for each formula that are computed using Algorithm~\ref{alg:boxes}.
The union of the box sets of all \stl formulae in $\set^\phi$ is represented as $\CA{B}_{\phi_{low}} \cup \CA{B}_{\phi_{high}} = \CA{B}_{\set^\phi}$ and is shown in Figure~\ref{fig:boxes}. 

\begin{figure}[tb]
\centering
\journal{
\subfloat[$\Proj(\phi_{low})$]{\includegraphics[width=0.40\linewidth]{philow.png}
\label{fig:philow}}
\subfloat[$\Proj(\phi_{high})$]{\includegraphics[width=0.4\linewidth]{phihigh.png}
\label{fig:phihigh}}\\
}
\subfloat[$\Proj(\phi_{low} \lor \phi_{high})$]{\includegraphics[width=0.4\linewidth]{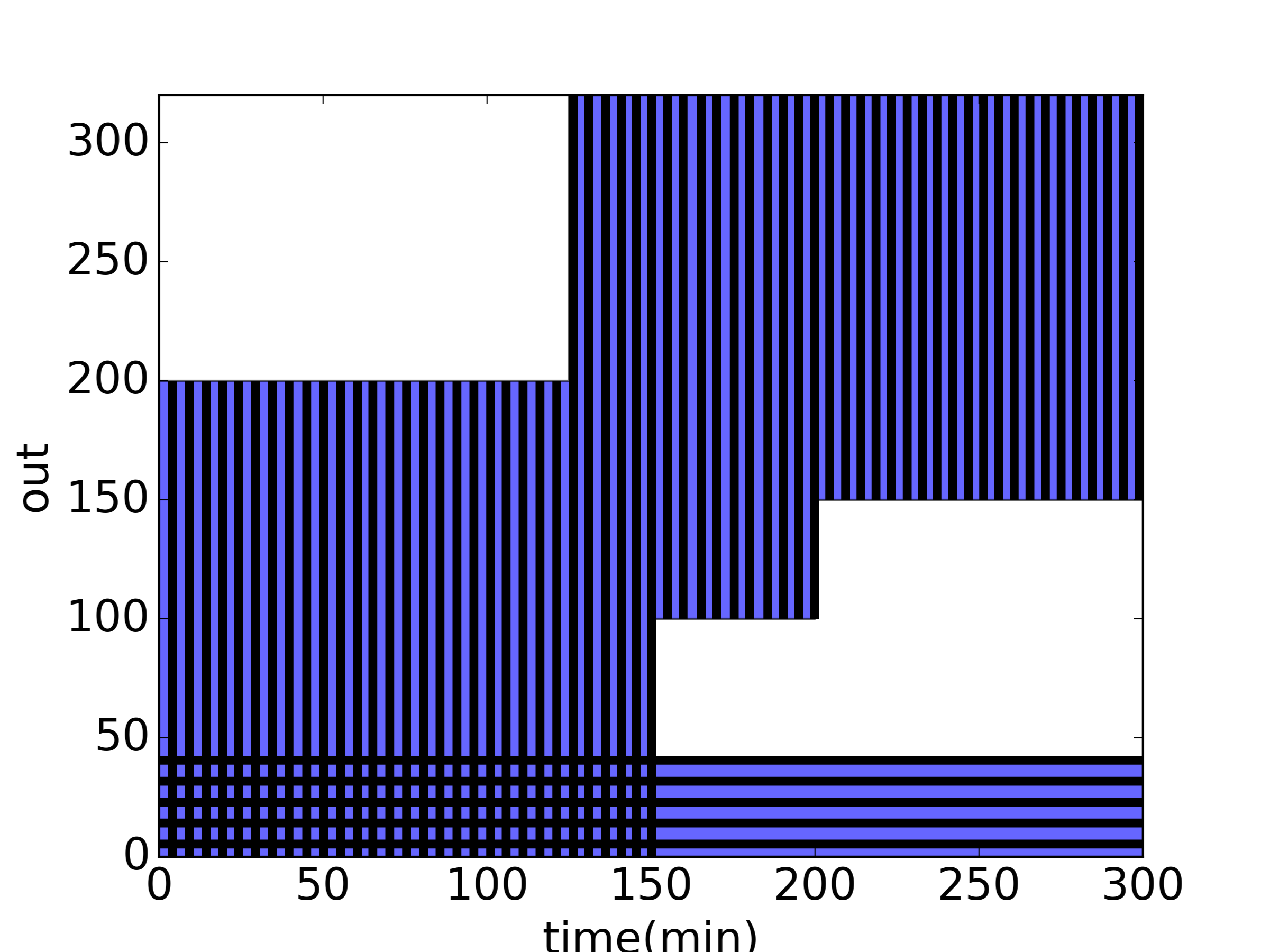}
\label{fig:union}}
\subfloat[$\CA{B}_{\phi_{low}} \cup \CA{B}_{\phi_{high}} = \CA{B}_{\set^\phi}$]{\includegraphics[width=0.4\linewidth]{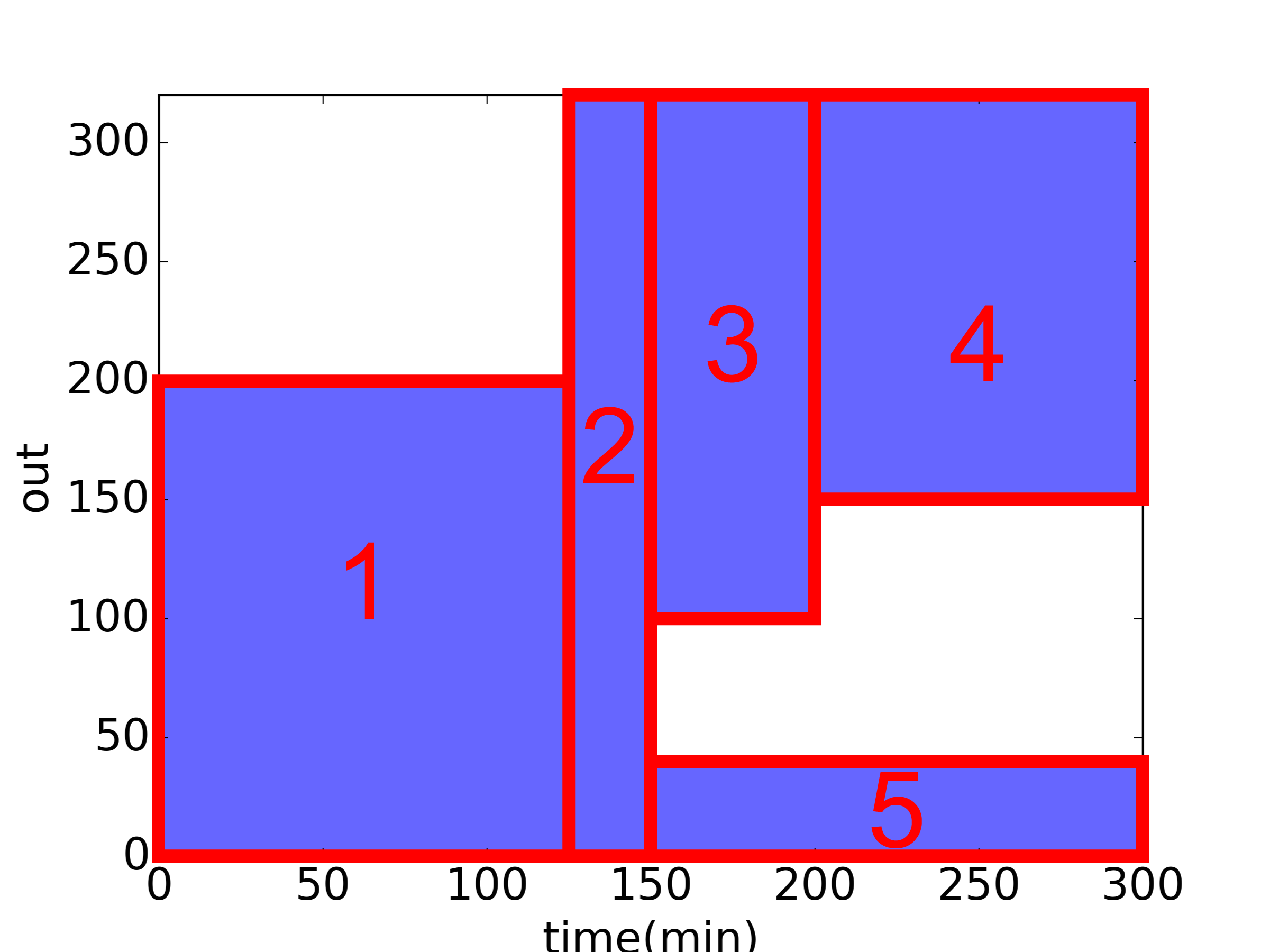}
\label{fig:boxes}}
\caption{\small \journal{\protect\subref{fig:philow} and \protect\subref{fig:phihigh} show the area of satisfaction for $\phi_{low}$ and $\phi_{high}$, respectively. The blue regions represent regions of satisfaction of the $\Always$ operator.} \protect\subref{fig:union} shows the union of the areas of satisfaction for $\phi_{low}$ and $\phi_{high}$. The vertical stripe and horizontal stripe areas represent $\Proj(\phi_{low})$ and $\Proj(\phi_{high})$, respectively. \protect\subref{fig:boxes} shows the boxes created for  $\CA{B}_{\phi_{low}} \cup \CA{B}_{\phi_{high}}$.}
\label{fig:twoClauseUnionBox}
\end{figure}

\journal {
Next, we use temporal logic inference using Grid TLI~\cite{vaidyanathan2017grid} to obtain an STL formula for the traces produced by each circuit. $\phi_{con}$ is the STL generated for the constitutive circuit while $\phi_{ind}$ is the STL generated for the induction circuit (illustrated in Figure~\ref{fig:gridtlirep}). 

\begin{figure}[!htb]
\centering
\subfloat[$\phi_{con}$]{\includegraphics[width=0.4\linewidth]{constitutiveTLI.png}
\label{fig:constitutiveTLI}}
\subfloat[$\phi_{ind}$]{\includegraphics[width=0.4\linewidth]{inductionTLI.png}
\label{fig:inductionTLI}}
\caption{\protect\subref{fig:constitutiveTLI} and \protect\subref{fig:inductionTLI} show the area of satisfaction of $\phi_{con}$ and $\phi_{ind}$, respectively. The area shaded in blue represents the $\Always$ operator.}
\label{fig:gridtlirep}
\end{figure}
}

Using Grid TLI~\cite{vaidyanathan2017grid}, we produce STL formulae, $\phi_{con}$ and $\phi_{ind}$, for each circuit using the traces shown in Figures~\ref{fig:constitutiveExperiment} and \ref{fig:inductionExperiment}, respectively.
We then use the \symdiff metric and get the following values: $d_{\symdiff}(\CA{B}_{\set^\phi},\CA{B}_{\phi_{con}})$ = $0.636$ and $d_{\symdiff}(\CA{B}_{\set^\phi},\CA{B}_{\phi_{ind}})$ = $0.304$. Using the \hausdorff metric, we get: $d_{\hausdorff}(\CA{B}_{\set^\phi},\CA{B}_{\phi_{con}})$ = $0.067$ and $d_{\hausdorff}(\CA{B}_{\set^\phi},\CA{B}_{\phi_{ind}})$ = $0$.
These results imply that the behavior of the induction circuit is closer to the desired specification than the circuit with constitutive expression, and thus, it should be selected as the desired circuit.




\section{Loss functions for TLI}
\label{sec:loss-tli}

Loss functions play a fundamental role in statistical inference and learning theory.
In this framework, we usually have a pair $(\CA{X}, \CA{Y})$ of real vector spaces corresponding
to the state and observation spaces, respectively.
Three ingredients are used in the formalization:
(a) a model of the states $p_X(x)$ -- the prior distribution,
(b) a model of the observations given the state $p_{Y\mid X}(y\mid x)$,
and (c) a real-valued loss function $L: \CA{X} \times \CA{X} \to \BB{R}$.
Let $h : \CA{Y} \to \CA{X}$ be a decision rule, which can also be
interpreted as a partition of
the state space $\CA{X}$ based on observations, and $\CA{H}$ be the set of all decision rules
or the hypothesis space.
The frequentist and Bayesian risks are defined based on the loss functions,
and induce optimal decision rules.
This general framework is the basis for the study and design of decision algorithms
in statistical inference and learning. For more details,
see~\cite{vapnik2013nature}.
\journal{
We can define the frequentist risk as
\begin{equation}
R_{freq}(x, h) = \BB{E}[L(x, h(y)) \mid X=x],
\end{equation}
where the expectation is with respect to $p_{X\mid Y}(x\mid y)$,
while the Bayesian risk is
\begin{equation}
R_{Bayes}(p_X, h) = \BB{E}[R_{freq}(x, h)],
\end{equation}
where $x \sim p_X(x)$.
We also have that $R_{freq}(y, h)=R_{Bayes}(\delta_y, h)$,
where $\delta_y$ is the Dirac distribution with all mass at $y$.
Lastly, we can define the optimal frequentist decision rule by the minimax rule
\begin{equation}
h^*_{freq} = \arg \min_{h \in \CA{H}} \max_{x \in X} R_{freq}(x, h),
\end{equation}
while the optimal Bayesian decision rule is 
\begin{equation}
h^*_{Bayes} = \argmin_{h \in \CA{H}} R_{Bayes}(p_X, h)
\end{equation}
for a given prior distribution $p_X$.
Clearly, the optimal decision rule depends on the chosen loss 
function, and this general framework is the basis for studying
and design of decision algorithms in statistical inference and
learning.
}

In the following, we show how we adapt this framework for TLI,
where we use the proposed metrics, \hausdorff and \symdiff,
as loss functions. In this paper, we only focus on the loss
functions, while characterization of optimal decision rules,
their computation, and regularization are left for future work.

For TLI, the state space $\CA{X} = \Phi^{\CA{S}_\maxhorz}$ is the
set of all time-bounded STL formulae, while the observation space
$\CA{Y} = \spow{\CA{S}_\maxhorz}$ is the set of all languages.
The hypothesis space $\CA{H}$ is composed of decision rules that
map languages to STL formulae.
Lastly, the loss functions are defined as
$L_{STL}: \Phi^{\CA{S}_\maxhorz} \times \Phi^{\CA{S}_\maxhorz} \to \BB{R}$
such that $L_{STL}(\phi, h(S))$ represents the dissimilarity
between the ground truth formula $\phi$ and the STL formula
obtained by the decision rule using the signal set $S \subseteq \CA{S}_\maxhorz$.
We propose to use the \hausdorff and the \symdiff metrics as loss
functions $L_{STL}$.

We assess the performance of the two decision rules from TLI: TreeTLI~\cite{bombara_decision_2016} based on decision trees,
and GridTLI~\cite{vaidyanathan2017grid} based on minimum covers of signals in space-time $\BB{S} \times [0, \maxhorz]$, with respect 
to the two metrics as shown in
Figure~\ref{fig:tli-stat-learn-results}.

The ground truth STL formula that was used to generate the signals
in Figure~\ref{fig:tli-stat-learn-data} is
\begin{align*}
\phi_{GT} =& \Always_{[0, 1]}((x_1 \leq 0.1) \andltl (x_2 \leq 0.6) \andltl (x_2 \geq 0.4)) \\
& \andltl \Always_{[7, 10]} ((x_1 \geq 0.7) \andltl ((x_2 \geq 0.8) \orltl (x_2 \leq 0.2)))
\end{align*}
where $\BB{S}=\BB{U}^2$ and $\maxhorz = \norm{\phi_{GT}} = 10$.
Note that TreeTLI requires both positive and negative examples,
while GridTLI only needs positive ones.
For brevity, we omit here the formalization for rules that require
both types of examples.

\begin{figure}[tb]
\centering
\subfloat[Positive and Negative Signals]{\includegraphics[height=0.43\linewidth]{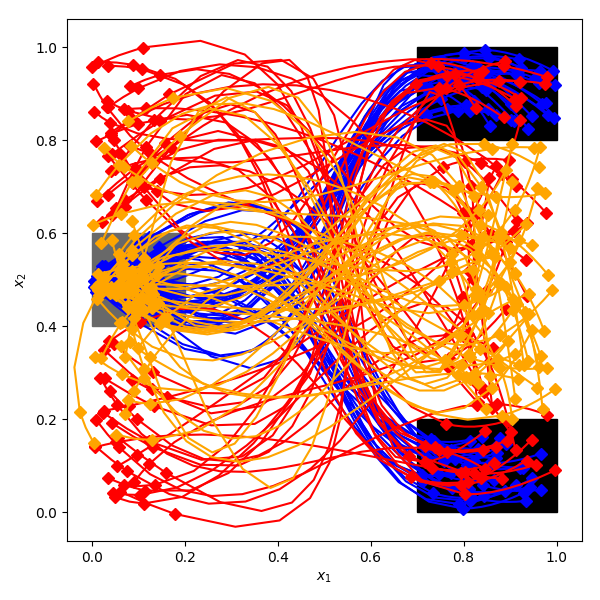}
\label{fig:tli-stat-learn-data}
}
\subfloat[Results]{\includegraphics[height=0.43\linewidth]{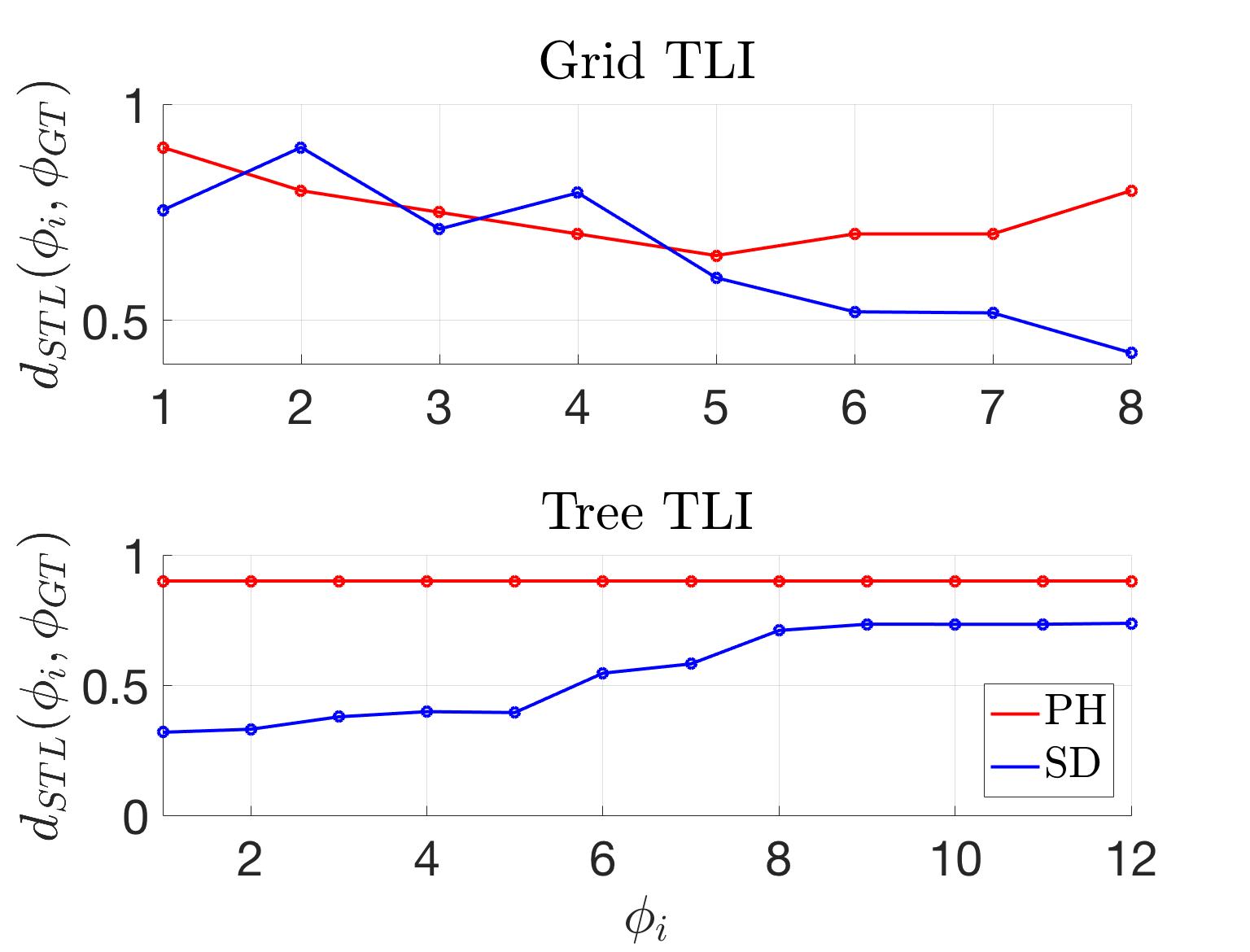}
\label{fig:tli-stat-learn-results}
}
\caption{\small \protect\subref{fig:tli-stat-learn-data} shows the blue positive, and red and orange negative example signals used by the two TLI algorithms.
The positive signals start in the gray region, and end in one of the black regions.
The red negative signals do not start in the gray region,
while the orange ones do not end in the black regions.
\protect\subref{fig:tli-stat-learn-results} shows the \hausdorff and the \symdiff distances
between the ground truth formula $\phi_{GR}$ and the learned formulae using
GridTLI and TreeTLI, respectively.
}
\end{figure}

The results in Figure~\ref{fig:tli-stat-learn-results} show
the distances between the ground truth formula $\phi_{GT}$
and the iterations of TreeTLI (lower plot) as the decision tree
grows~\cite{bombara_decision_2016}.
For GridTLI (upper plot), we varied the discretization thresholds~\cite{vaidyanathan2017grid}
from rougher to finer grids, $\delta_s \in \{0.5, 0.45, \ldots, 0.1\}$ and
$\delta_t \in \{5, 4.5, \ldots, 1\}$ for space and time, respectively.
The upper plot for GridTLI highlights the over-fitting phenomenon in the
\hausdorff metric (red), where reducing the discretization thresholds
helps reducing the error, but further reduction leads to over-fitting.
For the \symdiff (blue), the loss has a decreasing trend which we hypothesize
is due to a better temporal fitting that the \hausdorff distance does not capture.
In the case of TreeTLI (lower plot), the \hausdorff distance (red) is constant.
This masking behavior might be due to the compounding effect of
i) the primitives used do not match the structure of $\phi_{GT}$,
and ii) the incremental and local nature of TreeTLI.
Thus, the first step of the decision tree is heavily penalized by
the \hausdorff metric.
The \symdiff metric (blue), which shows an increasing trend,
is consistent with this conclusion.

Thus, the statistical learning approach to TLI gives insight into
the ability of algorithms to recover temporal logic rules assumed
to underlie data. It also provides a formal framework to study
TLI methods.
A detailed account of problems GridTLI and TreeTLI
are appropriate for based on the insights provided by the proposed
metrics is left for future work.


%% file: discussion.tex
\section{Discussion and Future Work}
\label{sec:discussion}

We presented two metrics for computing the distance of one STL formula to another.
These methods are very useful in applications where temporal logic specifications are mined from simulation or experimental data, and need to be compared against a desired specification.
Fields such as synthetic biology and robotics, where systems are characterized with performance specifications, can greatly benefit from our methods.
We also showed how these metrics are useful as a first step in evaluating the performance of TLI methods.

An immediate theoretical extension is studying continuous-time signals. By assuming Lipschitz continuity of signals, it is possible to provide bounds between the metrics computed in discrete-time and the ones in continuous-time. A similar idea was used in \cite{fainekos2009robustness} to compute sampled-time STL scores. The second extension is relaxing the assumption on rectangular predicates.  

%% file: STL_Comparison.bbl
\begin{thebibliography}{10}
\providecommand{\url}[1]{#1}
\csname url@samestyle\endcsname
\providecommand{\newblock}{\relax}
\providecommand{\bibinfo}[2]{#2}
\providecommand{\BIBentrySTDinterwordspacing}{\spaceskip=0pt\relax}
\providecommand{\BIBentryALTinterwordstretchfactor}{4}
\providecommand{\BIBentryALTinterwordspacing}{\spaceskip=\fontdimen2\font plus
\BIBentryALTinterwordstretchfactor\fontdimen3\font minus
  \fontdimen4\font\relax}
\providecommand{\BIBforeignlanguage}[2]{{%
\expandafter\ifx\csname l@#1\endcsname\relax
\typeout{** WARNING: IEEEtran.bst: No hyphenation pattern has been}%
\typeout{** loaded for the language `#1'. Using the pattern for}%
\typeout{** the default language instead.}%
\else
\language=\csname l@#1\endcsname
\fi
#2}}
\providecommand{\BIBdecl}{\relax}
\BIBdecl

\bibitem{baier2008modelchecking}
C.~Baier and J.~Katoen, \emph{Principles of model checking}.\hskip 1em plus
  0.5em minus 0.4em\relax {MIT} Press, 2008.

\bibitem{kress2009temporal}
H.~Kress-Gazit, G.~E. Fainekos, and G.~J. Pappas, ``Temporal-logic-based
  reactive mission and motion planning,'' \emph{IEEE transactions on robotics},
  vol.~25, no.~6, pp. 1370--1381, 2009.

\bibitem{batt2007robustness}
G.~Batt, B.~Yordanov, R.~Weiss, and C.~Belta, ``Robustness analysis and tuning
  of synthetic gene networks,'' \emph{Bioinformatics}, vol.~23, no.~18, pp.
  2415--2422, 2007.

\bibitem{coogan2017formal}
S.~Coogan, M.~Arcak, and C.~Belta, ``Formal methods for control of traffic
  flow: Automated control synthesis from finite-state transition models,''
  \emph{IEEE Control Systems}, vol.~37, no.~2, pp. 109--128, 2017.

\bibitem{Emerson1982}
E.~A. Emerson and E.~M. Clarke, ``Using branching time temporal logic to
  synthesize synchronization skeletons,'' \emph{Science of Computer
  Programming}, vol.~2, no.~3, pp. 241 -- 266, 1982.

\bibitem{pnueli1977temporal}
A.~Pnueli, ``The temporal logic of programs,'' in \emph{Foundations of Computer
  Science, 1977., 18th Annual Symposium on}.\hskip 1em plus 0.5em minus
  0.4em\relax IEEE, 1977, pp. 46--57.

\bibitem{maler2004monitoring}
O.~Maler and D.~Nickovic, ``Monitoring temporal properties of continuous
  signals,'' in \emph{Formal Techniques, Modelling and Analysis of Timed and
  Fault-Tolerant Systems}.\hskip 1em plus 0.5em minus 0.4em\relax Springer,
  2004, pp. 152--166.

\bibitem{fainekos2009robustness}
G.~E. Fainekos and G.~J. Pappas, ``Robustness of temporal logic specifications
  for continuous-time signals,'' \emph{Theoretical Computer Science}, vol. 410,
  no.~42, pp. 4262--4291, 2009.

\bibitem{donze2010robust}
A.~Donz{\'e} and O.~Maler, \emph{Robust satisfaction of temporal logic over
  real-valued signals}.\hskip 1em plus 0.5em minus 0.4em\relax Springer, 2010.

\bibitem{donze2013efficient}
A.~Donz{\'e}, T.~Ferrere, and O.~Maler, ``Efficient robust monitoring for
  {STL},'' in \emph{Computer Aided Verification}.\hskip 1em plus 0.5em minus
  0.4em\relax Springer, 2013, pp. 264--279.

\bibitem{Tumova-ACC13}
J.~Tumova, L.~Reyes-Castro, S.~Karaman, E.~Frazzoli, and D.~Rus,
  ``{Minimum-violating planning with conflicting specifications},'' in
  \emph{American Control Conference (ACC)}, 2013.

\bibitem{Tumova-HSCC13}
J.~Tumova, G.~C. Hall, S.~Karaman, E.~Frazzoli, and D.~Rus, ``Least-violating
  control strategy synthesis with safety rules,'' in \emph{International
  Conference on Hybrid Systems: Computation and Control}, Philadelphia, PA,
  USA, 2013, pp. 1--10.

\bibitem{VaTuKaBeRu-ICRA-2017}
C.-I. Vasile, J.~Tumova, S.~Karaman, C.~Belta, and D.~Rus, ``{Minimum-violation
  scLTL motion planning for mobility-on-demand},'' in \emph{IEEE International
  Conference on Robotics and Automation}, Singapore, Singapore, May 2017, pp.
  1481--1488.

\bibitem{VaAkBe-TCS-2016}
C.~I. Vasile, D.~Aksaray, and C.~Belta, ``{Time Window Temporal Logic},''
  \emph{Theoretical Computer Science}, vol. 691, no. Supplement C, pp. 27--54,
  August 2017.

\bibitem{kim2015minimalrevision}
K.~Kim, G.~Fainekos, and S.~Sankaranarayanan, ``On the minimal revision problem
  of specification automata,'' \emph{The International Journal of Robotics
  Research}, 2015.

\bibitem{ghosh2016diagnosis}
S.~Ghosh, D.~Sadigh, P.~Nuzzo, V.~Raman, A.~Donz{\'e}, A.~L.
  Sangiovanni-Vincentelli, S.~S. Sastry, and S.~A. Seshia, ``Diagnosis and
  repair for synthesis from signal temporal logic specifications,'' in
  \emph{Proc. International Conference on Hybrid Systems: Computation and
  Control}.\hskip 1em plus 0.5em minus 0.4em\relax ACM, 2016, pp. 31--40.

\bibitem{munkres2000topology}
J.~R. Munkres, \emph{Topology}.\hskip 1em plus 0.5em minus 0.4em\relax Prentice
  Hall, 2000.

\bibitem{conway2012course}
J.~B. Conway, \emph{A course in abstract analysis}.\hskip 1em plus 0.5em minus
  0.4em\relax American Mathematical Soc., 2012, vol. 141.

\bibitem{Bartocci2014}
E.~Bartocci, L.~Bortolussi, and G.~Sanguinetti, ``Data-driven statistical
  learning of temporal logic properties,'' in \emph{Formal Modeling and
  Analysis of Timed Systems}, A.~Legay and M.~Bozga, Eds.\hskip 1em plus 0.5em
  minus 0.4em\relax Cham: Springer International Publishing, 2014, pp. 23--37.

\bibitem{Jin2015}
X.~Jin, A.~Donzé, J.~V. Deshmukh, and S.~A. Seshia, ``Mining requirements from
  closed-loop control models,'' \emph{IEEE Transactions on Computer-Aided
  Design of Integrated Circuits and Systems}, vol.~34, no.~11, pp. 1704--1717,
  2015.

\bibitem{bombara_decision_2016}
G.~Bombara, C.-I. Vasile, F.~Penedo, H.~Yasuoka, and C.~Belta, ``A {{Decision
  Tree Approach}} to {{Data Classification Using Signal Temporal Logic}},'' in
  \emph{Proc. {{International Conference}} on {{Hybrid Systems}}:
  {{Computation}} and {{Control}}}.\hskip 1em plus 0.5em minus 0.4em\relax New
  York, NY, USA: {ACM}, 2016, pp. 1--10.

\bibitem{vaidyanathan2017grid}
P.~Vaidyanathan, R.~Ivison, G.~Bombara, N.~A. DeLateur, R.~Weiss, D.~Densmore,
  and C.~Belta, ``Grid-based temporal logic inference,'' in \emph{Annual
  Conference on Decision and Control (CDC)}.\hskip 1em plus 0.5em minus
  0.4em\relax IEEE, 2017, pp. 5354--5359.

\bibitem{Hoxha2018}
B.~Hoxha, A.~Dokhanchi, and G.~Fainekos, ``Mining parametric temporal logic
  properties in model-based design for cyber-physical systems,''
  \emph{International Journal on Software Tools for Technology Transfer},
  vol.~20, no.~1, pp. 79--93, Feb 2018.

\bibitem{Dokhanchi2014}
A.~Dokhanchi, B.~Hoxha, and G.~Fainekos, \emph{5th International Conference on
  Runtime Verification, Toronto, ON, Canada.}\hskip 1em plus 0.5em minus
  0.4em\relax Springer, 2014, ch. {On-Line Monitoring for Temporal Logic
  Robustness}, pp. 231--246.

\bibitem{Kraft2015}
D.~Kraft, ``Computing the {H}ausdorff distance of two sets from their signed
  distance functions,'' \emph{Computational Geometry}, Submitted March 2015.

\bibitem{raman2014model}
V.~Raman, A.~Donz{\'e}, M.~Maasoumy, R.~M. Murray, A.~Sangiovanni-Vincentelli,
  and S.~A. Seshia, ``Model predictive control with signal temporal logic
  specifications,'' in \emph{Decision and Control (CDC), 2014 IEEE 53rd Annual
  Conference on}.\hskip 1em plus 0.5em minus 0.4em\relax IEEE, 2014, pp.
  81--87.

\bibitem{vapnik2013nature}
V.~Vapnik, \emph{The nature of statistical learning theory}.\hskip 1em plus
  0.5em minus 0.4em\relax Springer, 2013.

\end{thebibliography}
